\documentclass[reqno]{amsart}

\usepackage{a4wide}
\usepackage{amsmath,amssymb,amsfonts,amsthm}
\usepackage{moreverb,rotating,graphics}
\usepackage[utf8]{inputenc}
\usepackage[english]{babel} 
\usepackage{framed,fancybox}
\usepackage{caption, subcaption}
\usepackage{mathrsfs, esint, dsfont}
\usepackage{eurosym}
\usepackage[colorlinks, linkcolor={blue}, citecolor={red}]{hyperref}
\usepackage{cancel}

\usepackage{caption}
\usepackage{float}
\usepackage{subcaption}

\newtheorem{theorem}{Theorem}[section]
\newtheorem{proposition}[theorem]{Proposition}
\newtheorem{remark}[theorem]{Remark}
\newtheorem{lemma}[theorem]{Lemma}

\newtheorem{example}[theorem]{Example}

\newcommand{\dps}{\displaystyle}
\newcommand{\RR}{\mathbb{R}}

\renewcommand{\SS}{\mathbb{S}}

\newcommand{\ri}{\mathrm{i}}
\newcommand{\cS}{\mathscr{S}}
\newcommand{\cD}{\mathcal{D}}
\newcommand{\cE}{\mathcal{E}}

\newcommand{\cC}{\mathcal{C}}
\newcommand{\cR}{\mathcal{R}}

\newcommand{\R}{\mathbb{R}}
\newcommand{\N}{\mathbb{N}}
\newcommand{\Z}{\mathbb{Z}}

\newcommand\bk{{\bold k}}

\newcommand\bR{{\bold R}}

\newcommand\bx{{\bold x}}
\newcommand\by{{\bold y}}

\newcommand\bnull{{\bold 0}}

\def\cC{{\mathcal C}}
\def\cD{{\mathcal D}}
\def\cE{{\mathcal E}}
\def\cF{{\mathcal F}}
\def\cG{{\mathcal G}}
\def\cH{{\mathcal H}}
\def\cI{{\mathcal I}}

\def\cP{{\mathcal P}}

\def\cR{{\mathcal R}}
\def\cS{{\mathcal S}}

\def\fS{{\mathfrak S}}

\def\rd{{\mathrm{d}}}
\def\re{{\mathrm{e}}}
\def\ri{{\mathrm{i}}}

\newcommand\1{{\ensuremath {\mathds 1} }} 

\newcommand{\TF}{{\rm TF}}
\newcommand{\LT}{{\rm LT}}
\newcommand{\rHF}{{\rm rHF}}

\newcommand{\Tr}{{\rm Tr}}

\newcommand{\VTr}{\underline{\rm Tr}}

\newcommand{\norm}[1]{\left\| #1\right\|}
\newcommand{\set}[1]{\left\{ #1\right\}}
\newcommand{\bra}[1]{\left( #1\right)}
\newcommand{\av}[1]{\left| #1\right|}
\newcommand{\com}[1]{\left[ #1\right]}

\newcommand\ii{{\infty}}

\newcommand{\review}[1]{#1}

\title[DFT for 2-d homogeneous materials]{Density Functional Theory for two-dimensional homogeneous materials}

\author{David Gontier}
\address[David Gontier]{CEREMADE, University of Paris-Dauphine, PSL University, 75016 Paris, France}
\email{gontier@ceremade.dauphine.fr}

\author{Salma Lahbabi}
\address[Salma Lahbabi]{ESSM, LRI, ENSEM, UHII, 7 Route d’El Jadida, B.P. 8118 Oasis, Casablanca; 
MSDA, Mohammed VI Polytechnic University, Lot 660, Hay Moulay Rachid Ben Guerir, 43150,
	Morocco}
\email{s.lahbabi@ensem.ac.ma}

\author{Abdallah Maichine}
\address[Abdallah Maichine]{MSDA, Mohammed VI Polytechnic University, Lot 660, Hay Moulay Rachid Ben Guerir, 43150,
	Morocco}
\email{abdallah.maichine@um6p.ma}

\date{\today}

\begin{document}

\begin{abstract}
    We study Density Functional Theory models for systems which are translationally invariant in some directions, such as a homogeneous $2$-d slab in the $3$-d space. We show how the different terms of the energy are modified and we derive reduced equations in the remaining directions. In the Thomas--Fermi model, we prove that there is perfect screening, and provide decay estimates for the electronic density away from the slab. In Kohn--Sham models, we prove that the Pauli principle is replaced by a penalization term in the energy. In the reduced Hartree-Fock model in particular, we prove that the resulting model is well-posed, and give some properties of the minimizer.
    
    
    
\end{abstract}

\maketitle    

\tableofcontents


\section{Introduction}

Density Functional Theory (DFT) was first introduced~\cite{Hohenberg1964, Kohn1965} to study the 
quantum energy of finite systems, such as molecules. It became an important tool also in condensed matter physics, to study infinite systems, such as crystals. Using thermodynamic limit procedure~\cite{Fefferman1985}, it was shown that the energy {\em per unit cell} of a crystal could be computed as the minimization of an explicit (periodic) functional. This was first proved for the Thomas--Fermi model in~\cite{Catto1998_book}, and for the reduced Hartree-Fock model in~\cite{Catto2002}.

In this work, we consider the intermediate case, and study semi-infinite systems, where the system is infinite in $s$ directions, but is localized in $d$ other directions. One can think {\em e.g.} of a nano-wire in three-dimensional space (corresponding to $s = 1$ and $d = 2$), or an infinite slab (corresponding to $s = 2$ and $d = 1$). 
Our interest comes from the recent developments of two-dimensional materials, such as graphene and phosphorene in the physics community~\cite{VdW}. Such systems, have been studied in \review{\cite{Blanc_LeBris_00} in the framework of Thomas--Fermi type models and in}~\cite{Lingling1,Lingling2} in the framework of the reduced Hartree-Fock model. Our objective is to prove and study DFT models for reduced dimension systems in order to allow low computational cost. 

\medskip\noindent

In this work, we focus on the simple case where the system is {\em homogeneous} in its first $s$ variables, and we derive reduced equations in the remaining $d$ variables. More specifically, we consider a positive charge distribution $\mu \ge 0$ in $\R^{s+d}$, representing the semi-infinite homogeneous material, which is translation invariant in the first $s$ dimensions:
\begin{equation} \label{eq:form_mu}
\mu(x_1, x_2, \cdots, x_s, x_{s+1}, \cdots, x_{s+d}) = \mu(0, \cdots, 0, x_{s+1}, \cdots, x_{s+d}).
\end{equation}
We are interested in the state and the (normalized) energy of the electrons in the potential generated by this charge distribution. 

While this hypothesis may look too simple and unrealistic, the model highlights some important features for semi-infinite systems. It may not reproduce the correct physical properties of a real life two-dimensional material, since it does not take into account the microscopic details; however, we think that the effect of these microscopic details fade away from the slab exponentially fast, so that our model reproduces the correct behavior of the density far from the slab. This fact will be the object of future works.

\medskip

Let us briefly explain our main results. We focus on the Thomas--Fermi model and the reduced Hartree-Fock model for simplicity, although our derivation works for general Kohn--Sham models. Also, in what follows, we consider a two-dimensional slab ($s = 2$) in three dimensions ($d = 1$, so that $s + d = 3$), but the techniques apply similarly to other cases.

In the Thomas--Fermi model, the energy only depends on the electronic density $\rho$, which we assume share the same invariance as $\mu$ in~\eqref{eq:form_mu}. Starting from the full three-dimensional periodic Thomas--Fermi energy, we can define an {\em energy per unit surface}, which takes the form
\[
    \cE^\TF(\rho) := c_\TF \int_\R \rho^{5/3} + \frac12 \cD_1(\rho - \mu).
\]
Here, $\cD_1$ is the one-dimensional Hartree energy, that we define below. We minimize the energy $\cE^\TF$ among positive densities $\rho : \R \to \R^{{+}}$ satisfying $\int_\R \rho = \int_\R\mu$. The advantage of the new problem is that the dimension has been reduced: the problem is set on the real line $\R$ only (instead of the full space $\R^3$). We study this reduced model, and prove various properties such as the existence and uniqueness of a minimizer and the perfect screening of dipolar moments. We also prove Sommerfeld estimates, which states that $\rho(x)$ decays as $| x |^{-6}$ away from the slab. 

In the reduced Hartree-Fock model, we prove similarly that one can reduce the dimension of the problem. Starting from the three-dimensional problem set on one-body density matrices $\gamma$ satisfying the Pauli principle $0 \le \gamma \le 1$, we obtain a one-dimensional problem set on self-adjoint operators $G$ acting on $L^2(\R)$, which are postive: $G \ge 0$, but which {\em no longer need to satisfy the Pauli principle $G \le 1$}. Instead, the Pauli principle appears as a penalization term in the energy: the three-dimensional kinetic energy $\frac12 \Tr_3( - \Delta \gamma)$ is replaced by a one-dimensional kinetic energy of the form
\[
    \frac12 \Tr_1({-\Delta}G) + \pi \Tr_1(G^2).
\]
This last term, sometime called the Tsallis entropy, prevents $G$ from having large eigenvalues. \review{In other words, we rigorously derive the Tsallis entropy: it appears here as a weak form of the Pauli principle, coming from the collapse of some dimensions}. The reduced energy then takes the form
\[
    \cE^\rHF(G) := \frac12 \Tr_1({-\Delta}G) + \pi \Tr_1(G^2) + \cD_1(\rho_G - \mu).
\]
\review{Although we restrict ourselves to the reduced Hartree-Fock model for simplicity, similar derivations can be performed for general Kohn--Sham models (see the comments following Theorem~\ref{th:gamma_G}).}

In this paper, we justify both models, and study the existence, uniqueness and properties of the minimizers (denoted by $\rho_\TF$ and $G_*$ respectively). We also provide numerical simulations for the two models. To our surprise, we found out that, with high accuracy, we have $\rho_\TF \approx \rho_{G^*}$. In particular, the Thomas--Fermi model is a very good approximation of the reduced Hartree-Fock model after this reduction of the dimensionality.

\review{As we already mentioned, we believe that these simple toy models do reproduce the correct behavior at infinity of the density and mean-field potential. If true, this would give some indications on how to numerically simulate such slab systems, and in particular on how large the size of the simulation box should be.}

The paper is structured as follows. We state our main results in Section~\ref{sec:main_results}. We prove the results concerning the Thomas--Fermi model in Section~\ref{sec:TF}, and the ones for the reduced Hartree-Fock model in Section~\ref{sec:rHF}. We explain in particular the regularization of the one-dimensional Hartree term in Section~\ref{ssec:Hartree}. Our numerical illustrations are gathered in Section~\ref{sec:numerical}. Finally, we provide in Appendix~\ref{sec:LT} a simple proof of a Lieb-Thirring type inequality, which uses  techniques similar to the one developed in this work.

\subsection*{Acknowledgments} This work has received fundings from a CNRS international cooperation program ({\em Projet International de Collaboration Scientifique}, or PICS, of D.G. and S.L.). The research leading to these results has received funding from OCP grant AS70 ``Towards phosphorene based materials and devices''.


\section{Main results}
\label{sec:main_results}

Let us explain in more details our results. 

\subsection{Main results in the Thomas--Fermi model}

In orbital-free models, such as the Thomas--Fermi (TF) model~\cite{thomas1927calculation, fermi1927metodo, vonWeizsacker1935theorie}, the energy depends solely on the electronic density $\rho$. We refer to~\cite{Lieb1977} for a mathematical study of this model in the molecular case. 

The nuclear density $\mu$ is a positive function which is $\RR^2$-translation invariant, {\em i.e.} it is of the form
$$\mu(x_1, x_2,x_3) = \mu(0, 0, x_{3})= \mu(x_3)\in L^1(\RR).$$
We denote by $Z:=\int_\RR \mu  > 0$ the charge per unit surface, which is not necessarily an integer. 

In order to derive the reduced TF equation, we perform a simple 'thermodynamic limit'. We set the problem on the tube $\Gamma_L := [-\frac{L}{2}, \frac{L}{2}]^2 \times \R$ (with periodic boundary conditions), and we consider neutral systems:
$$
\rho\in L^1(\Gamma_L),\quad \int_{\Gamma_L}\rho=L^2Z.
$$
The supercell TF energy of such a density is
\begin{equation} \label{eq:def:ELTF}
    \cE_L^\TF(\rho) = c_{\rm TF} \int_{\Gamma_L} \rho^{5/3} + \frac12 \cD_{3, L}(\rho - \mu). 
\end{equation}
The first term is the Thomas--Fermi kinetic energy, with $c_{\rm TF} = \frac{3}{10}(3 \pi^2)^{2/3}$ the usual three-dimensional Thomas--Fermi constant, and the second term is the supercell three-dimensional Hartree term, describing  the interaction of the electrons with the charge density $\mu$, as well as a mean-field self-interaction of the electrons. The quadratic form $\cD_{3, L}$ (the subscript $3$ refers to the space dimension $s+d = 3$) is formally defined by
$$
\cD_{3, L}(f) := \iint_{(\Gamma_L)^2} G_L( \bx- \by) f (\bx) f (\by)  \rd \bx \rd \by,
$$
where $G_L$ is the $L$-periodic Green's function solution to
\begin{equation} \label{eq:def:GL}
    - \Delta G_L = 4 \pi\sum_{(R_1, R_2) \in L \Z^2} \delta_{(R_1, R_2, 0)},
\end{equation}
with the periodic boundary conditions $G_L(x_1 + R_1, x_2 + R_2, x_3) = G_L(x_1, x_2, x_3)$ for all $(R_1, R_2) \in L \Z^2$ and the symmetry condition $G_L(x_1, x_2, -x_3) = G_L(x_1, x_2, x_3)$. 

As we prove in Proposition~\ref{prop:TF-equivalence}, this model is related to the one-dimensional TF  energy
\begin{equation} \label{eq:1d_TF}
    \boxed{ \cE^{\rm TF}(\rho) =  c_{\rm TF} \int_{\R} \rho^{5/3} + \frac12 \cD_1(\rho - \mu).}
\end{equation}
This energy is interpreted as an {\em energy per unit surface}. It has the three-dimensional Thomas--Fermi power $5/3$ for the kinetic energy, but the one-dimensional Hartree term, formally defined by
\begin{equation}\label{eq:def-D1}
    \cD_1(f) := - 2 \pi \iint_{\R \times \R} | x- y | f(x ) f( y) \rd x \rd y .
\end{equation}
This definition of $\cD_1$ is only valid for functions that decay fast enough. We give a regularization of this expression suitable for neutral functions in Section~\ref{ssec:1dHartree}. 

The energy~\eqref{eq:def:ELTF} is set on a three-dimensional space $\Gamma_L$, while the energy~\eqref{eq:1d_TF} is set on the line~$\R$. The latter is easier to study both theoretically and numerically. Our first proposition shows that the minimization problems concerning the two energies are equivalent.

\begin{proposition}\label{prop:TF-equivalence} Consider the minimization problems
$$
        \cI^\TF_L := \inf \left\{  \cE_L^\TF(\rho), \quad \rho \in L^1(\Gamma_L)\cap L^{5/3}(\Gamma_L), \; \rho \ge 0,\; \cD_{3, L}(\rho-\mu)<\ii,\; \int_{\Gamma_L} \rho = L^2 Z \right\}
 $$
 and
$$
    \cI^\TF := \inf \left\{ \cE^\TF(\rho), \quad  \rho \in L^1(\R)\cap L^{5/3}(\R), \;\ \rho \ge 0,\; \cD_1(\rho-\mu)<\ii,\; \int_{\R} \rho = Z \right\}.
$$
    Then, for all $L > 0$, we have $\cI^\TF_L = L^2 \cI^\TF$, and both energies share the same minimizer, which depends only on $x_3$.

\end{proposition}
The proof can be read in Section~\ref{ssec:reduction-TF}. This justifies the reduced one-dimensional problem~\eqref{eq:1d_TF}. We now focus on this reduced model~\eqref{eq:1d_TF}, and prove that it is well-posed, in the sense that it admits a unique minimizer.

\begin{theorem}\label{th:mainTF}
The infinimum $\cI^\TF$ is finite and admits  a unique minimizer $\rho_{\rm TF}$. It is the unique solution to the Thomas--Fermi equation
  \begin{equation}\label{eq:TF-EL}
      \begin{cases}
          \frac{5}{3} c_{\rm TF}  \rho_\TF^{2/3}  = \left( \lambda - \Phi_\TF \right)_+, \\
          -\Phi_\TF'' = 4 \pi (\rho_\TF - \mu),\quad \Phi_\TF'(\pm \infty) = 0 \quad \text{and}\quad \Phi_\TF(0)=0,
      \end{cases}
  \end{equation}
   where $x_+ := \max \{ 0, x\}$, and where $\lambda \in \R$ is chosen so that $\int_\R \rho_\TF = Z$. Here, $\Phi_{\TF}$ is the mean-field potential, defined as the unique solution of the second equation.
\end{theorem}

The proof of Theorem~\ref{th:mainTF} is presented in Section~\ref{ssec:TF_existence}
and the properties of the mean-field potential $\Phi_\TF$ are detailed in Section~\ref{ssec:TF_properties}. Uniqueness comes from the strict convexity of the energy $\cE^\TF$. We recall that, in dimension one, there is no reference energy (the one-dimensional Green's function does not have a limit as $x \to \infty$), so all potentials and all Fermi levels are defined up to global constants. Only the difference $\Phi_\TF - \lambda$ has a physical meaning, and is called the mean-field Thomas--Fermi potential.

\medskip

Next we study the screening property of the Thomas--Fermi model. We prove that the dipolar moment of $\mu$ is perfectly screened by the TF density $\rho_\TF$. Recall that if $f$ is any function with $\int_\R f = 0$ (neutral), then the potential generated by $f$ is formally given by
\[
    \Phi_f(x)  := - 2 \pi \int_{\R} f(y) | x - y | \rd y.
\]
If $f$ is compactly supported, say in $[-R, R]$, then we have
\[
    \forall x > R, \quad \Phi_f(x) =  2 \pi \int_\R f(y) y \rd y, \quad \text{while} \quad
    \forall x < -R, \quad \Phi_f(x) =  -2 \pi \int_\R f(y) y \rd y.
\]
The {\em dipolar} moment of $f$ is defined as the difference 
\[
    \Phi_f(\infty) - \Phi_f( - \infty) = 4 \pi \int_\R f(y) y \rd y.
\]

The next result states that all dipolar moments are perfectly screened in the Thomas--Fermi model (see Section~\ref{ssec:screening_TF} for the proof).

\begin{proposition}[Screening of dipolar moments]\label{prop:properties_PhiTF}
Assume that $\mu$ satisfies $ | x | \mu(x) \in L^1(\RR)$. Then, the density $\rho_\TF$ satisfies $ | x | \rho_\TF(x) \in L^1(\RR)$ as well, the potential $\Phi_\TF$ is continuous, bounded on $\R$,  satisfies  $\Phi_\TF \leq \lambda$  and 
    \begin{equation} \label{eq:limPhiinfty}
        \lim_{x \to + \infty} \Phi_\TF(x) = \lim_{x \to - \infty} \Phi_\TF(x) = \lambda.
    \end{equation}    
In particular, 
\begin{equation} \label{eq:perfect_screening_TF}
    \int_{\R} x \bra{ \rho_\TF(x) - \mu(x) }\rd x = 0. 
\end{equation}
\end{proposition}

Finally, we prove explicit decay rates of the density $\rho_\TF$ in the case where $\mu$ is compactly supported. The following result can be seen as a one-dimensional version of the Sommerfeld estimates~\cite{sommerfeld1932asymptotische, solovej2003ionization}. In particular, the decay rates of $\rho_\TF$ and $\Phi_\TF$ are independent of the charge $Z$ of the system.

\begin{proposition}[Sommerfeld estimates for a slab] \label{prop:sommerfeld}
    Assume that $\mu$ is a  compactly supported with ${\rm Supp}(\mu) \subset [a, b]$. Then there exist $x_a,x_b\in\RR$ such that 
    $$ 
    \begin{cases}
        \Phi_\TF(x) - \lambda   = - \frac{c_1}{(x-x_a)^4} \text{ for } x\leq a, \\
        \Phi_\TF(x) - \lambda   = - \frac{c_1}{(x-x_b)^4} \text{ for } x\geq b,
    \end{cases}
    \quad \text{and} \quad
    \begin{cases}
         \rho_\TF(x)  = \frac{c_2}{(x-x_a)^6}\text{ for } x\leq a, \\
          \rho_\TF(x)  = \frac{c_2}{(x-x_b)^6}\text{ for } x\geq b.
    \end{cases}
$$
  with the constants
  \[
    c_1 := \frac{5^5c_\TF^3}{27\pi^2}
    \quad \text{and} \quad
    c_2 := \frac{5^6c_\TF^3}{27\pi^3}.
 \] 
\end{proposition}

The proof is provided in Section~\ref{ssec:Sommerfeld}. With the usual Thomas--Fermi constant, $c_\TF = \frac{3}{10}(3 \pi^2)^{2/3}$, we find
$$
c_1=\frac{5^3 3^2}{8} \pi^2\quad \text{and} \quad
c_2 =\frac{5^43^2}{8} \pi.
$$
This result is reminiscent of the usual Sommerfeld estimate for atoms~\cite{sommerfeld1932asymptotische, solovej2003ionization}, where similar results hold, but with the different constants $c_1 = \frac{3^4}{8} \pi^2$ and $c_2 =  \frac{3^5}{8} \pi$.

The previous result allows to have an explicit solution for the Dirac case $\mu = Z \delta_0$ (perfect charged slab).

\begin{example}[The Dirac case]
    In the case where $\mu = Z \delta_0$, the solution is explicit. By convexity, the optimal density $\rho_\TF$ is even. Together with Proposition~\ref{prop:sommerfeld}, we deduce that it is of the form $\rho_\TF(x) = c_2 (| x | + \alpha)^{-6}$. The value of $\alpha$ can be found using {the fact} that $\int \rho_\TF=Z$. This gives
    \[
    Z = 2\int_0^\infty \dfrac{c_2}{(| x | + \alpha)^6} \rd x = \dfrac{2 c_2}{\alpha^5} \int_0^\infty \dfrac{1}{( 1 + y )^6} \rd y = \dfrac{2  c_2}{5 \alpha^5}.
    \]
    Hence
    \[
    \boxed{ \rho_\TF(x) = c_2 \left( | x |   + \left( \tfrac{2 c_2}{5 Z} \right)^{1/5} \right)^{-6}.}
    \]
\end{example}

\subsection{Main results for the rHF model}\label{ssec:main-rHF}
We now turn to the reduced Hartree-Fock model. For the sake of simplicity, we work with spinless electrons,
but our arguments can be extended to models with spin.
In this model, the state of the electrons is described by a one-body density matrix $\gamma$, which is a self-adjoint operator satisfying the Pauli principle
\begin{equation} \label{eq:PauliPrinciple}
	0 \le \gamma \le 1.
\end{equation}
We consider a setting similar to the previous section, that is, the nuclear charge distribution $\mu$ is a positive function that is invariant under translations in the first two variables, and integrable in the third one. 

For $\cH$ a Hilbert space, we denote by $\cS(\cH)$ the set of self-adjoint operators acting on $\cH$. In a tube $\Gamma_L := [ - \frac{L}{2};\frac{L}{2}]^2 \times \R$, the three-dimensional rHF energy of a one-body density matrix $\gamma \in \cS \left(L^2(\Gamma_L)\right)$ is of the form
\begin{equation} \label{eq:def:ErHF_L}
\cE_L^{\rHF}(\gamma) = \frac12 \Tr_{3,L} \left( - \Delta \gamma  \right) + \frac12 \cD_{3, L}(\rho_\gamma - \mu),
\end{equation}
where the last term is the Hartree energy, as in the Thomas--Fermi model, and $\Tr_{3, L}(- \Delta \gamma)$ represents the supercell kinetic energy of $\gamma$. Here and thereafter, the subscripts $3$ or $1$ refer to the space dimension.

The energy~\eqref{eq:def:ErHF_L} needs to be minimized over all density matrices $\gamma$ satisfying the Pauli principle $0 \le \gamma \le 1$, and the neutrality condition $\Tr_{3, L}(\gamma) = ZL^2$. By convexity of the functional, and translation invariance in the first two variables, it is enough to consider density matrices which commute with these translations. Such operators have kernels which satisfy
\begin{equation} \label{eq:assumption_gamma}
	\gamma(x_1, x_2, x_3 ; y_1, y_2, y_3) = \gamma(x_1 - y_1, x_2 - y_2, x_3 ;  0, 0, y_3) =: \gamma(\bx -\by; x_3, y_3).
\end{equation}
\review{We denote by  $\cP$ the set  one-body density matrices $\gamma \in \cS(L^2(\RR^3))$ that satisfy~\eqref{eq:assumption_gamma}, the Pauli principle, and that have a finite trace per unit surface
$$
\VTr_3( \gamma):= \frac{1}{L^2}\Tr_{3, L} \left( \1_{\Gamma_L} \gamma \1_{\Gamma_L}  \right),\quad \forall L\in \RR^+_* $$
(see Section~\ref{ssec:trv} for more details).  The normalized energy (per unit surface) of a state $\gamma\in \cP$ is then 
\begin{equation}\label{eq:energy-rHF-3D}
\cE^{\rHF}_3(\gamma)= \frac12 \VTr_3(-\Delta \gamma) +\frac12 \cD_1(\rho_\gamma-\mu),
\end{equation}
where the quadratic form $\cD_1$ was introduced in~\eqref{eq:def-D1}. 
}
\medskip

The energy $\cE^\rHF_3$ depends on $\gamma$, which is an operator acting on a three-dimensional space.  In order to have an energy depending on an operator acting on a one-dimensional space, we use the following key result, whose proof can be read in Section~\ref{ssec:reduction-rHF}.

\begin{theorem} \label{th:gamma_G}
	For any $\gamma \in \cP$, there is an operator $G$ in 
	\[
	\cG := \left\{ G \in \cS(L^2(\R)), \; G \ge 0, \; \Tr(G) < \infty  \right\},
	\]
	satisfying $\rho_G = \rho_\gamma$ (same density), and
	\begin{equation} \label{eq:ineq:kineticGamma_G}
		\frac12 \VTr_{3}(- \Delta \gamma) \ge \frac12 \Tr_1 \left( - \Delta G \right) + \pi\Tr_1 \left( G^2  \right).
	\end{equation}
	Conversely, for any $G \in \cG$, there is $\gamma \in\cP $ so that $\rho_\gamma = \rho_G$, and for which there is equality in~\eqref{eq:ineq:kineticGamma_G}. \review{In particular, for any (representable) density $\rho$,
    \[
        \inf \left\{ \frac12  \VTr_{3}(- \Delta \gamma), \ \gamma \in \cP, \ \rho_\gamma = \rho \right\} = 
        \inf \left\{ \frac12 \Tr_1 \left( - \Delta G \right) + \pi\Tr_1 \left( G^2  \right), \ G \in \cG, \ \rho_G = \rho \right\}.
    \]}
\end{theorem}

Theorem~\ref{th:gamma_G} allows us to prove that the problem set on three-dimensions coincides with a problem set on the real line. \review{Actually, following the constrained-search approach by Levy and Lieb~\cite{Levy1979, Lieb1983}, we see that for a general Kohn--Sham model of the form $\cE_3^{\rm KS}(\gamma) = \frac12 \VTr_{3}(- \Delta \gamma) + E^{\rm Hxc}(\rho_\gamma)$ ($E^{\rm Hxc}$ is the Hartree exchange-correlation energy per unit surface), we formally have
\begin{align*}
    \inf_{\gamma}   \cE^{\rm KS}_3(\gamma) 
    & = \inf_{\rho} \left\{  \inf_{\gamma \to \rho} \left\{ \frac12  \VTr_{3}(-\Delta \gamma) \right\} + E^{\rm Hxc}(\rho)  \right\} \\
    & = \inf_{\rho} \left\{ \inf_{G \to \rho} \left\{ \frac12 \Tr_1(-\Delta G) + \pi \Tr(G^2) \right\} +  E^{\rm Hxc}(\rho)  \right\} = \inf_G \cE^{\rm KS}(G),
\end{align*}
with a reduced energy per unit surface of the form $\cE^{\rm KS}(G) := \frac12 \Tr_1(-\Delta G) + \pi \Tr_1(G^2) + E^{\rm Hxc}(\rho_G)$. In the above computation, we restricted the minimization to one-body density matrices satisfying the $\RR^2$-translation invariance condition~\eqref{eq:assumption_gamma}. In general Kohn--Sham models, which are non-convex, symmetry breaking may happen, and the optimal $\gamma$ may not be translationally invariant. This is why we restrict ourselves to the (convex) reduced Hartree-Fock model in the sequel and consider the one-dimensional energy per unit surface}
\begin{equation} \label{eq:def:ErHF}
    \boxed{ \cE^\rHF(G) := \frac12 \Tr_1 \left( - \Delta G \right) + \pi \Tr_1 \left( G^2  \right) + \frac12 {\cD_1}(\rho_G - \mu). }
\end{equation}

The fact that the three-dimensional model~\eqref{eq:energy-rHF-3D} and the one-dimensional model~\eqref{eq:def:ErHF} are equivalent is stated in the next Theorem, whose proof is postponed to Section~\ref{ssec:reduction-rHF}.

\begin{theorem} \label{th:rHF_equivalence}
Consider the minimization problems
\[
    \cI^\rHF_3:= \inf\set{\cE_3^\rHF(\gamma),\quad  \gamma\in\cP \quad {\VTr}_{3}(\gamma)=Z}
\]
and
\[
    \cI^\rHF := \inf\set{\cE^\rHF(G),\quad   G \in \cG, \;\Tr_1(G)  = Z  }.
\]
Then $\cI^\rHF_3 = \cI^\rHF$ and the minimizers of both energies share the same density, which depends only on $x_3$.
\end{theorem}

Compared to~\eqref{eq:def:ErHF_L}, the energy~\eqref{eq:def:ErHF} is one-dimensional, which can be efficiently studied both theoretically and numerically.  Note that in the reduced problem, there is no Pauli condition for the operator $G$. It is somehow replaced by the penalty term $\pi \Tr_1 (G^2)$ in the energy, which prevents $G$ from having large eigenvalues. \review{This term is strictly convex, and is sometimes called the Tsallis entropy. Our result shows that this term can be interpreted as an effective Pauli principle, coming from a collapse of some dimensions. We believe that Theorem~\ref{th:gamma_G} can be applied in various situations. For instance, we show in Appendix~\ref{sec:LT} how to use it to obtain a Lieb-Thirring type inequality.} 


\medskip\noindent
We now focus on the reduced problem. We first prove that it is well-posed (see Section~\ref{ssec:rHF_is_well_posed} for the proof).

\begin{theorem} \label{th:main_rHF}
	The infimum~$\cI^\rHF$ is finite and admits a unique minimizer $G_* \in \cG$. This minimizer satisfies the Euler-Lagrange equations
    \[
        \begin{cases}
            \dps G_* = \tfrac{1}{2 \pi} \left( \lambda - H_*  \right)_+ \\
            \dps H_* := - \tfrac12 \Delta + \Phi_*, \\
             -\Phi_*'' = 4 \pi (\rho_* - \mu),\quad \Phi_*'(\pm \infty) = 0 \quad \text{and}\quad \Phi_*(0)=0,
        \end{cases}
    \]
    where $\lambda \in \R$ is the Fermi level chosen such that $\Tr(G_*) = Z$,  $\rho_*:=\rho_{G_*}$ is the associated density to $G_*$ and $\Phi_*$ is the mean-field potential, defined as the unique solution of the last equation.
\end{theorem}

The operator $H_*$ is the mean-field one-body Schrödinger operator. The uniqueness of the minimizer comes from the strict convexity of the $\cE^\rHF$ functional, thanks to the $\Tr(G^2)$ term. Compared to the {\em usual} Euler-Lagrange equations for the {\em usual} reduced Hartree--Fock model, we see that the optimizer if of the form $G = (\lambda - H)_+$ instead of $\gamma = \1 \left( \lambda - H > 0 \right)$. This regularization comes from the Tsallis entropy $\Tr(G^2)$ term in the energy. One consequence is that, since the map $\lambda \mapsto \Tr_1(\lambda - H)_+$ is {\em strictly} increasing on $[\inf \sigma (H), \infty)$, the Fermi level $\lambda \in \R$ is always uniquely defined.

Compared to the Thomas--Fermi case, one cannot say much on the screening properties of the rHF model. Still, we record the following (see Section~\ref{ssec:props-V-rHF} for the proof).

\begin{proposition} \label{prop:Phi_rHF}
    Assume that $\mu$ satisfies $| x | \mu(x) \in L^1(\R)$. Then the density $\rho_*$ of $G_*$ satisfies $| x | \rho_*(x) \in L^1(\R)$ as well. The potential $\Phi_*$ is continuous, bounded on $\R$, and satisfies
    \[
        \lim_{x \to \pm \infty} \Phi_*(x) =  \pm 4 \pi \int_{\R^\pm} x \left(\rho - \mu \right)(x) \rd x.
    \]
\end{proposition}

It is unclear that the limits of $\Phi_*$ at 
$+\infty$ and $-\infty$  are equal, which would imply as before the perfect screening of the dipolar moment. Still, our numerical simulations in Section~\ref{sec:numerical} seem to indicate that, even though this term may not be null, it is always very small. In the special case where $\mu$ is even, $\rho_*$ is even as well by convexity, and $\Phi_*(-\infty) = \Phi_*({+}\infty)$: there is no dipolar moment in this case.

\review{In the Thomas--Fermi case, we were able to prove that the density $\rho_\TF$ decays as $| x |^{-6}$ far from the slab whenever $\mu$ is compactly supported (see Proposition~\ref{prop:sommerfeld}). Unfortunately, we were not able to fully characterized the decay of the density in the reduced Hartree-Fock case. The proof of the next result relies on a Bargmann type bound (see Section~\ref{ssec:proof_rhoG}).
\begin{proposition} \label{prop:rhoG}
    Assume that $\mu$ satisfies $| x |^3 \mu(x) \in L^1(\R)$. Then, if $| x |^3 \rho_*(x) \in L^1(\R)$ as well, then $G_*$ is finite rank, and $\rho_*$ is exponentially decaying away from the slab.
\end{proposition}
We think that if $\mu$ decays fast enough, then we always have $| x |^3 \rho_*(x) \in L^1(\R)$. Unfortunately, we were not able to prove this fact. This would imply that $G_*$ is always finite rank, and that $\rho_*$ is always exponentially decaying.}


\begin{remark}[General dimension] \label{rem:general_case_inequality}
	In the general case, where $\gamma$ is an operator acting on $\R^{s+d}$, which is translation invariant with respect to its first $s$--variables, we obtain similar results. Let us emphasize the differences. In Theorem~\ref{th:gamma_G}, the operator $G$ now acts on the last $d$ variables $G \in \cS(L^2(\R^d))$, and~\eqref{eq:ineq:kineticGamma_G} is replaced by
	\[
	\frac12 \VTr_{s+d}( - \Delta_{s+d} \gamma) \ge \frac12 \Tr_d \left( - \Delta_d G\right) + c_{\TF}(s) \ \Tr_d \left( G^{1 + \frac{2}{s}} \right),
	\]
	with the (spinless) Thomas--Fermi constant
	\[
	c_\TF(s) :=  \frac{s}{s+2} \left(\frac{s}{\left| \SS^{s-1} \right|}\right)^{2/s} 2 \pi^2.
	\]
    In Theorem~\ref{th:main_rHF}, the first Euler-Lagrange equation takes the form
	\[
	G_* = \dfrac{s}{(s + 2)} \dfrac{1}{c_\TF(s)} \left(\lambda - H_* \right)_+^{s/2} .
	\]
\end{remark}

\section{Homogeneous 2-d  materials in the  Thomas--Fermi model}
\label{sec:TF}

We start by proving some properties of the three- and one-dimensional Hartree energy. 

\subsection{Hartree interaction}
\label{ssec:Hartree}

\subsubsection{Three-dimensional Hartree interaction}

Let us first give an explicit expression of the three-dimensional Green's function $G_L$ defined in~\eqref{eq:def:GL}.

\begin{proposition}\label{prop:3D-hartree}
    Denoting by $\bx = (x_1, x_2) \in \R^2$ the first two variables, we have
    \[
        G_L(\bx, x_3) =  - \frac{2 \pi }{L^2} | x_3 | +  \frac{ 2 \pi }{L^2}
        \sum_{\bk  \in  \bra{\frac{2 \pi}{L}\Z}^2 \setminus \{ \bnull \}} 
        \frac{\re^{ - \av{\bk} \av{x_3}}}{\av{\bk}}
        \re^{\ri \bk \cdot \bx} + c ,
    \]
    where $c\in\RR$. 
    In particular, if $f(\bx, x_3) = f(\bnull, x_3)$ depends only on the last variable $x_3$, then
    \[
        \int_{\Gamma_L} f(\by, y_3) G_L(\bx - \by, x_3 - y_3) \rd \by \rd y_3 =  - 2 \pi \int_{\R} f(y_3) | x_3 - y_3 | \rd y_3  + c \int_{\R} f(x_3) \rd x_3.
    \]
\end{proposition}

\begin{proof}
    Thanks to the periodicity of $G_L$ in the first  two variables, we can make the following Ansatz 
    $$
    G_L(\bx, x_3) = \sum_{ \bk \in \left(\tfrac{2 \pi}{L} \Z \right)^2} c_\bk(x_3) \re^{ \ri \bk \cdot \bx}.
    $$
    Using the two-dimensional Poisson formula in the plane $(x_1, x_2)$
    \[
    \sum_{\bR \in L \Z^2} \delta_\bR(\bx) = \dfrac{1}{L^2} \sum_{\bk \in \left(\tfrac{2 \pi}{L} \Z \right)^2} \re^{\ri \bk \cdot \bx},
    \]
    the equation defining $G_L$ in~\eqref{eq:def:GL} becomes
    \[
        \forall \bk \in \left( \tfrac{2 \pi}{L } \Z \right)^2, \quad -  c_\bk''(x_3) + | \bk |^2 c_\bk(x_3)  = \frac{4 \pi}{L^2} \delta_0(x_3).
    \]
    For $\bk = \bnull$, the above equation gives $- c_\bnull''(x_3) = \frac{4 \pi}{L^2} \delta_0(x_3)$, whose general solution is 
    \[
        c_\bnull(x_3) = -2 \pi \frac{| x_3 |}{L^2} + \lambda x_3 + c,\quad c\in\RR.
    \] 
    By symmetry, we get $\lambda = 0$. For $\bk \neq \bnull$, we solve the equation on $(-\infty,0)$ and $(0,+\infty)$ and obtain that
   \[
        \forall x_3 > 0, \quad c_{\bk}(x_3) =  \alpha_\bk^+ \re^{-| \bk | \cdot |x_3|}, 
        \quad \text{and} \quad
        \forall x_3 < 0, \quad c_{\bk}(x_3) =  \alpha_\bk^- \re^{-| \bk | \cdot |x_3|}.
   \]
   	The symmetry condition implies that $c_\bk(x_3)=c_\bk(-x_3)$, so $\alpha_\bk^+=\alpha_\bk^-$ and $c_{\bk}(x_3) =  \alpha_\bk \re^{-| \bk | \cdot |x_3|}$. We find the value of $\alpha_\bk$ by looking at the singularity at $x_3 = 0$. We must have
    \[
        - c_\bk''  + | \bk |^2 c_\bk = 2 | \bk | \alpha_\bk \delta_0=\frac{4 \pi}{L^2} \delta_0(x_3),
    \]
    so $\alpha_\bk = \frac{2 \pi}{L^2 | \bk |}$, and the result follows.
\end{proof}

Away from the slab $x_3 = 0$, $G_L$ is exponentially close to $-\frac{2\pi}{L^2}\av{x_3}$. The Green's function $G_L$ is defined up to a global constant $c$, that we take equal to $0$ for simplicity. Actually, for neutral system, which is our main interest here, the choice of the constant is irrelevant.

\subsubsection{One-dimensional Hartree interaction}\label{ssec:1dHartree}

We are now interested in the one-dimensional Hartree term. We recall that it is formally given by (we put a tilde here to emphasize that we will soon consider another definition)
\[
    \widetilde{\cD_1}(f) := - 2 \pi \iint_{\R \times \R}  | x - y | f( x ) f ( y)\rd x \rd y.
\]
\begin{remark}
    In the literature, one often considers the one-dimensional Green's function $G_1$, solution to
    \[
    - \Delta G_1 = | \SS^0 | \delta, \quad \text{with} \quad | \SS^0 | = 2,
    \]
    whose solution is the usual $G_1(x) = - | x |$. In this article, our Green's function is rather $G = 2 \pi G_1$, and satisfies $-\Delta G = | \SS^2 | \delta_0$ with $|\SS^2| = 4 \pi$.
\end{remark}

One problem with this expression is that it is not well-defined if $f$ has a slow decay at infinity. Another problem is that the map $f \mapsto \widetilde{\cD_1}(f)$ is not convex in general.
In our context, we consider the Hartree interaction of $f := \rho - \mu$, which has a null integral. We therefore adopt the following definition. 

For $f \in L^1(\R)$, we set
$$
    W_f(x)=\int_{-\ii}^x f(y)\rd y, \quad \text{and} \quad
    \widehat{f}(k):=  \dfrac{1}{(2 \pi)^{1/2}} \int_\R  \re^{ - \ri k x} f(x)\rd x.
$$
The function $W_f$ is a primitive of $f$. We define the regularized version of the one-dimensional Hartree term by
\begin{equation} \label{eq:def:D1}
\cD_1(f) := 4\pi \int_{\R} \frac
{| \widehat{f}(k) |^2}{\av{k}^2}\rd k = 4\pi \int_{\RR}\av{W_f(x)}^2 \rd x.
\end{equation}
This expression is well defined for $f$ in the Coulomb space
$$
f \in \cC := \set{  f \in L^1(\RR),\; W_f \in L^2(\RR)}= \set{  f \in L^1(\RR),\; \frac{\widehat{f }}{\av{\cdot}}\in L^2(\RR)}.
$$
Due to the singularity at $k = 0$, any $f \in \cC$ must be neutral, in the sense
$$
    W(+\infty) = \int_\RR f(x)\rd x =\widehat{f}(0)=0.
$$

The two definitions $\widetilde{\cD_1}$ and $\cD_1$ do not coincide in general: $\cD_1$ is defined for neutral functions, while $\widetilde{\cD_1}(f)$ makes sense whenever $f$ decays sufficiently fast at infinity. The next Proposition shows that they coincide for neutral functions $f$ which decays fast enough.

\begin{proposition} 
\label{prop:1D-hartree}
The following holds.
   \begin{enumerate}
       \item  The map $\cC \ni f \mapsto \cD_1(f)$ is strictly convex.
       \item If $f \in L^1(\RR)$ satisfies $\int_\R f= 0$ and $\av{x}f(x)\in L^1(\RR)$, then $f \in \cC$, and
       \[
       \widetilde{\cD_1}(f) = \cD_1(f).
       \]
       \item  If $f \in \cC$, then we have
       \[
       \cD_1(f) = 4 \pi \iint_{(\R^+)^2 \cup (\R^-)^2} \min \{ | x |, | y | \} f(x) f(y)  \rd x \rd y =
       \int_{\R} \Phi_f(x)  f(x) \rd x ,
       \]
       where $\Phi_f(x) := -4\pi \int_0^x W_f (y)\rd y $ is the mean-field potential, also given by
       \begin{equation} \label{eq:def:Vf}
           \Phi_f(x) := 4 \pi \int_{\R^\pm} \min \{ | x |, | y | \} f(y) \rd y  \quad \text{for} \quad x \in \R^\pm.
       \end{equation}
       The function $\Phi_f$ is continuous, and is the unique solution to
       \[
        - \Phi_f''(x) = 4 \pi f, \quad \Phi_f'(x) \xrightarrow[x \to \pm \infty]{}0, \quad \Phi_f(0) = 0.
       \]
   \end{enumerate}
\end{proposition}

\begin{proof}
    The first point comes from the Fourier representation of ${\cD_1}$, which involves a strictly positive kernel $| k |^{-2}$. Let us prove the other  two points. Let $f\in L^1( \RR)$ be such that $\av{x} f(x)\in L^1(\RR)$ and $\int_\RR f = 0$. First, we see that
    \[
        \widehat{f}(k) = \widehat{f}(0) + ( \widehat{f} ) '(0) k + o(k)
     \]
     with
     \[
        \widehat{f}(0)=\frac{1}{(2 \pi)^{1/2}} \int_{\R} f(x)  \rd x = 0, \quad \text{and} \quad
        (\widehat{f})'(0) = \frac{1}{(2 \pi)^{1/2}} \int_{\R} x f(x)  \rd x \,\in \R.
    \]
    This proves that $\frac{\widehat{f}(k)}{| k |}$ is indeed in $L^2(\R)$ (there is no singularity at $k = 0$), so $f \in \cC$.
Besides, we have for $x\geq 0$, using that $\int_\R f = 0$ and Fubini,
\begin{align}
 \Phi_f(x)&=-4\pi \int_0^x \int_{-\ii}^y f(t)\rd t \,\rd y= 4\pi \int_0^x \int_y^{+\ii} f(t)\rd t  \,\rd y=4\pi\int_{(\RR^+)^2}\1_{y\leq\min\set{ x,t}} f(t)\rd t  \,\rd y \nonumber \\
 &=4\pi\int_{\RR^+}\min\set{ x,t}f(t)\rd t =4\pi\bra{\int_0^x t f(t) \rd t  + x\int_x^{+\ii}f(t)\rd t }. \label{eq:expression_Phi}
\end{align}
The last equality is somehow a one-dimensional version of Newton's theorem. A similar equality holds for $x\leq 0$. A similar computation shows that
$$
\cD_1(f)=4\pi \int_\RR \av{W_f}^2=4 \pi \iint_{(\R^+)^2 \cup (\R^-)^2} \min \{ | x |, | y |\}\, f(x)f(y)\,\rd x\,\rd y.
$$
Therefore 
$$
\cD_1(f) = \int_\R \Phi_f(x) f(x) \rd x.
$$
Finally, to prove that this expression is also $\widetilde{\cD_1}(f)$ for $f\in\cC$ satisfying $\av{x}f(x)\in L^1(\RR)$, we remark that
    \[
        | x | + | y | - | x - y | = \begin{cases}
            2\min \{\av{ x}, \av{y }\} & \quad \text{on} \quad (\RR^+)^2\cup(\RR^-)^2 \\
            0 & \quad \text{otherwise}.
        \end{cases}
    \]
    This gives
    \begin{align*}
        \cD_1(f) & = 2 \pi \int_{\R^2} \left( | x | + | y | - | x - y |  \right) f(x) f(y) \rd x \rd y \\
        &=        2 \pi \int_\R | x | f(x) \rd x \int_\R f(y) \rd y + 2 \pi \int_\R | y | f(y) \rd y  \int_\R f(x) \rd x + \widetilde{\cD_1}(f),
    \end{align*}
    and the  first two terms vanish since $\int_\R f = 0$ and $\av{x}f(x)\in L^1(\RR)$.
    \end{proof}

\subsection{Reduction of the Thomas Fermi model: Proof of Proposition~\ref{prop:TF-equivalence}}\label{ssec:reduction-TF}

In this section, we prove Proposition~\ref{prop:TF-equivalence}: we justify that the three-dimensional Thomas--Fermi problem equals its one-dimensional version. Recall that we defined
\[
    \cE_{3,L}^\TF(\rho) := c_{\rm TF} \int_{\Gamma_L} \rho^{5/3} + \frac12 \cD_{3,L}(\rho - \mu) 
    \quad \text{and} \quad
    \cE^\TF(\rho) := c_{\rm TF} \int_{\R} \rho^{5/3} + \frac12 \cD_1(\rho - \mu).
\]
We also introduce 
\[
    \widetilde{\cE^\TF}(\rho) := c_{\rm TF} \int_{\R} \rho^{5/3} + \frac12 \widetilde{\cD_1}(\rho - \mu).
\]
Let $\rho : \Gamma_L \to \R^+$ be a test three-dimensional density. We define
 $$
\widetilde{ \rho}(x_3):=\frac{1}{L^2}\int_{\com{-\frac{L}{2},\frac{L}{2}}^2}{\rho}(\bx,x_3)\,\rd \bx. 
 $$
 By convexity of the $\cE_{3,L}^\TF$ functional, we have $\cE^{\TF}_{3, L}(\widetilde{\rho}) \le \cE^\TF_{3, L}(\rho)$. In addition, using Proposition~\ref{prop:3D-hartree}, we see that 
 \begin{align*}
 \cD_3(\widetilde{\rho}-\mu)&=-2\pi\int_{\com{-\frac{L}{2},\frac{L}{2}}^2}\int_\RR\int_{\RR} (\widetilde{\rho}-\mu)(x_3)(\widetilde{\rho}-\mu)(y_3)\av{x_3-y_3} \,\rd y_3 \,\rd x_3  \\
 &=L^2 \widetilde{\cD_1}(\widetilde{\rho}-\mu).
 \end{align*}
 Therefore  $\cE_{3,L}^\TF({\rho}) \geq  \cE_{3,L}^\TF(\widetilde{\rho}) = L^2 \widetilde{\cE^\TF}(\widetilde{\rho})$. On the other hand, if $\rho : \R \to \R^+$ is a one-dimensional density, one can extend $\rho$ in the three dimensions setting $\rho(x_1, x_2, x_3) := \rho(x_3)$, and we have $\cE_{3,L}^\TF(\rho)  = L^2 \widetilde{\cE^\TF}(\rho)$. This proves that $\inf \cE_{3,L}^\TF = L^2 \inf \widetilde{\cE^\TF}$ and both energies share the same minimizer.

\subsection{Existence of minimizers: Proof of Theorem~\ref{th:mainTF}}
\label{ssec:TF_existence}

As we said before, the problem with the $\widetilde{\cD_1}$ Hartree term turns out to be quite difficult to study. 
In what follows, we rather study the problem $\cE^\TF$ instead of $\widetilde{\cE^\TF}$, that is with the regularized $\cD_1$ Hartree term instead of $\widetilde{\cD_1}$. Still, we prove in this section that if $| x | \mu(x) \in L^1(\R)$, then the optimal density $\rho$ satisfies $| x | \rho(x) \in L^1(\R)$ as well. In particular, for this density, we have $\cE^\TF(\rho)= \widetilde{\cE^\TF}(\rho)$.

\medskip\noindent

We therefore focus on the one-dimensional Thomas--Fermi minimization problem
\begin{equation*} 
\inf \set{\cE^{\rm TF}(\rho),\; \rho\in \cR}, \quad \text{with} \quad
\cR := \left\{ \rho \in L^1(\R) \cap L^{5/3}(\R), \quad \rho - \mu \in \cC, \quad \rho \ge 0 \right\},
\end{equation*}
and we prove Theorem~\ref{th:mainTF}.

\subsubsection{Existence and uniqueness of minimizer}
We first prove that the problem is well posed, and admits a unique minimizer.

We start by noting that ${\cR}$ is not empty: for instance, we have $\frac{1}{2}\mu*{e^{-\av{\cdot}}}\in{\cR}$. Since $\cD_1$ is a positive quadratic form on ${\cC}$, the energy functional $\cE^\TF$ is positive on $\cR$, thus bounded from below.  Let $(\rho_n)_n$ be a minimizing sequence in $\cR$. In particular, $(\rho_n)_n$ is bounded in $L^{1}(\RR)\cap L^{5/3}(\RR)$ and $(W_{\rho_n-\mu})_n$ is bounded in $L^2(\R)$. Up to sub-sequences, there exist $\rho\in L^{1}(\R)\cap L^{5/3}(\RR)$ and $W\in L^2(\R)$ such that $\rho_n\rightharpoonup\rho$ and $W_{\rho_n-\mu}\rightharpoonup W$ weakly in $L^{1}(\R)\cap L^{5/3}(\RR)$  and $L^2(\R)$ respectively. Let us prove that $W = W_{\rho - \mu}$. For a test function $\psi\in C_c^\infty(\R)$, we have
\begin{align*}
\langle W,\psi'\rangle &=\lim_{n\to\infty}\langle W_{\rho_n-\mu},\psi'\rangle
 = -\lim_{n\to\infty} (\rho_n-\mu,\psi) = -\langle\rho-\mu,\psi\rangle.
\end{align*}
We deduce that $W' = \rho-\mu$ in the distributional sense, so $W(x)=W_{\rho-\mu}(x) + c$ for some constant $c$. Since $W\in L^2(\R)$, we have $c=\lim_{-\infty}W(x)=0$ hence $W = W_{\rho-\mu}$ as wanted. This implies the neutrality condition $\int_\R \rho(x) \rd x=\int_\R \mu(x) \rd x =Z$, so $\rho\in \cR$.  By the lower semi-continuity of the $L^{5/3}(\R)$ and the $L^2(\R)$ norms, we obtain 
$$
\cE^\TF(\rho)\leq \liminf \cE^\TF(\rho_n)=\inf \set{\cE^\TF(\rho),\; \rho\in \cR}. 
$$
Hence $\rho$ is a minimizer. Uniqueness follows from the strict convexity of the $\cE^\TF$ functional.

\subsubsection{The Euler-Lagrange equations}
In what follows, we denote by $\rho_\TF$ the optimal density. We prove in this section that $\rho_\TF$ satisfies the Euler-Lagrange equations~\eqref{eq:TF-EL}. First, we have
    \[
    \forall h \in C_0^\infty(\R), \ \int_\R h = 0, \ \rho + h \ge 0, \  \forall t \in [0, 1], \quad  {\cE^\TF}(\rho_\TF + t h)  \ge {\cE^\TF}(\rho_\TF).
    \]
    Differentiating at $t = 0$ and using that
    \[
        \frac12 \cD_1(\rho_\TF + t h - \mu) = \frac12 \cD_1(\rho_\TF - \mu) + t \int_\R \Phi_{\TF} h + o(t^2), 
    \]
    where $\Phi_\TF := \Phi_{\rho_\TF - \mu}$, we obtain
    \begin{equation} \label{eq:firstTFeqt}
         \forall h \in C_0^\infty(\R), \ \int_\R h = 0, \ \rho + h \ge 0, 
         \quad
         \int_\R \left(  \frac{5}{3}  c_\TF  \rho_\TF^{2/3}(x) + \Phi_\TF(x) \right) h(x) \rd x \ge 0.
    \end{equation}
    As in~\cite{Lieb1977}, we see that on the set $\{x\in \RR, \;  \rho(x) > 0\}$, $h$ can locally takes positive and negative values, so, on this set, we must have $ \frac{5}{3}  c_\TF  \rho_\TF^{2/3}(x) + \Phi_\TF(x) = \lambda$ for some $\lambda \in \R$, called the Fermi level. In particular, we have $\Phi_\TF < \lambda$ on this set. On the set $\{x\in \RR, \; \rho(x) = 0\}$, $h$ can  take only positive values, and we deduce that $\Phi_\TF \ge \lambda$. This gives the usual Thomas--Fermi equation
    \begin{equation} \label{eq:eqt_rho_Phi}
        \frac{5}{3}  c_\TF  \rho_\TF^{2/3}(x)  = \left[ \lambda - \Phi_\TF(x) \right]_+,
    \end{equation}
    where $[ f]_+ := \max \{ 0, f\}$. 
    The same reasoning as in~\cite{Lieb1977} shows that if a density satisfies the TF equation~\eqref{eq:eqt_rho_Phi}, then it is the unique minimizer of the TF energy functional. 
    
    Since $\Phi_\TF$ is continuous, the density $\rho_\TF$ is also continuous. Let us prove that $\Phi_\TF \le \lambda$. We recall the following maximum principle in one-dimension.

    \begin{lemma}\label{lemma:principe-maximum}
        Let $V : \R \to \R $ be a continuous function such that: 
        \begin{itemize}
            \item for any $x$ such that $V(x)\geq 0$, we have $V''(x)\geq 0$,
            \item $V'\to 0$ at $\pm\ii$.
        \end{itemize}
        Then $V\leq 0$ or $V$ is constant.
    \end{lemma}
    Before proving Lemma~\ref{lemma:principe-maximum}, we show how to use it to conclude that $\Phi_\TF \le \lambda$. We set $V:=\Phi_\TF-\lambda$. For any $x$ such that $V(x)\geq 0$, we have $\rho_\TF(x)=0$ from~\eqref{eq:eqt_rho_Phi}, hence $V''(x)= \Phi_\TF''(x) = -4\pi(\rho(x)-\mu(x))=4\pi\mu(x)\geq0$. Besides $V'=\Phi_\TF'\to 0$ at $\pm\ii$ by Proposition~\ref{prop:1D-hartree}. Thus $V$ satisfies the conditions of the lemma and we conclude that either $V$ is constant or that $V=\Phi_\TF - \lambda\le 0$. If $V$ is constant, then so is $\rho_\TF$ and since $\rho_\TF$ is integrable, then $\rho_\TF=0$, which is not possible as $\int_{\R}\rho_{\TF}=Z>0$. We conclude that $\Phi_\TF\leq \lambda$. 
    
    Now, the Euler-Lagrange equation can be written as
    $$
        \frac{5}{3}  c_\TF  \rho_\TF^{2/3}(x)  +\Phi_\TF(x) = \lambda . 
    $$
    It remains to provide the:
    \begin{proof}[Proof of Lemma~\ref{lemma:principe-maximum}]
        Let us assume that there is $x_0\in\RR$ such that $V(x_0)>0$. Let $x_m\in \RR\cup\set{-\ii}$ and $x_M\in \RR\cup\set{+\ii}$ defined by
        $$
            x_m=\sup\set{x\leq x_0, \; V(x)\leq 0}\quad \text{and}\quad x_M=\inf\set{x\geq x_0, \; V(x)\leq 0}.
        $$
        By continuity of $V$, the open interval  $I=(x_m,x_M)$ is not empty. On this interval we have $V(x)>0$ thus $V'' (x) \geq 0$. There are 4 possibilities:
        \begin{enumerate}
            \item $x_m,x_M\in\RR$. In this case $V(x_m)=V(x_M)=0$ by continuity. As $V$ is convex on $I$, it follows that $V\leq 0$ on $I$ by the maximum principle, a contradiction;
            \item $x_m\in \RR$ and $x_M=+\ii$. In this case $V(x_m)=0$ by continuity. On $I$, $V''\geq 0$, thus $V'$ is non decreasing. Besides $V'\to 0$ at $+\ii$. Therefore $V'\leq 0$ on $I$. As $V(x_m)=0$, it follows that $V\leq 0$  on $I$, a contradiction;  
            \item $x_m=-\ii$ and $x_M\in\RR$. This case is treated as the previous one;
            \item $x_m=-\ii$ and $x_M=+\ii$. In this case $V''\geq 0$, thus $V'$ is non decreasing on $\RR$. However, $V'\to 0$ at $\pm \ii$ then $V'=0$, that is $V$ is constant.
        \end{enumerate}
        We conclude that $V\leq 0$ on $\RR$ or is constant. 
    \end{proof}

\subsection{Properties of the TF density and mean-field potential}
\label{ssec:TF_properties}

 In this section, we give some extra properties of $\rho_\TF$ and $\Phi_\TF$.

\subsubsection{Screening of dipolar moments: Proof of Proposition~\ref{prop:properties_PhiTF}}
\label{ssec:screening_TF}

In what follows, we assume that $| x | \mu(x) \in L^1(\R)$ (the first moment of $\mu$ is finite). Let us prove that $| x | \rho_\TF(x) \in L^1(\R)$ as well. We have, by Equation~\eqref{eq:def:Vf}, for $x \ge 0$, that
\begin{equation} \label{eq:formula_PhiTF}
    \frac{1}{4 \pi} \Phi_\TF(x) = \int_0^\infty (\rho_\TF - \mu)(y) \min \{ x, y \} \rd y 
    = \int_0^x y (\rho_\TF - \mu)(y) \rd y + x \int_x^\infty (\rho_\TF - \mu)(y) \rd y .
\end{equation}
This gives
\[
    \int_0^x y \rho_\TF(y) \rd y \le \frac{1}{4 \pi} \Phi_\TF(x) +  \int_0^x y \mu(y) \rd y + x \int_x^\infty \mu (y) \rd y
    \le \frac{\lambda}{4\pi} + 2 \int_0^\infty | y | \mu(y) \rd y,
\]
where we used that $x \mu(y) \le y\mu(y)$ for $y \ge x$. This proves that $\int_{\R^+} x \rho(x) < \infty$. We can prove a similar result for $x \le  0$, which proves $| x | \rho_\TF(x) \in L^1(\R)$.

In particular, we have, for $x \ge 0$,
\[
    x \int_x^\infty (\rho_\TF - \mu)(y) \rd y \le \int_x^\infty y  (\rho_\TF - \mu)(y) \rd y \xrightarrow[x \to +\infty]{} 0.
\]
Together with~\eqref{eq:formula_PhiTF}, we obtain
\[
    \lim_{x \to \infty} \Phi_\TF(x) = 4 \pi  \int_{\R^+} y  (\rho_\TF - \mu)(y) \rd y 
    \quad \text{and} \quad
     \lim_{x \to -\infty} \Phi_\TF(x) = - 4 \pi  \int_{\R^-}  y  (\rho_\TF - \mu)(y) \rd y.
\]
In particular, $\Phi_\TF$ have limits at $\pm \infty$. Moreover, since $\Phi_\TF$ is continuous, we deduce that $\Phi_\TF$ is bounded. In addition, by the Euler Lagrange equation~\eqref{eq:TF-EL}, we have $\rho_\TF(x) \to (\lambda- \Phi_\TF(\pm \infty))^{2/3}$ as $x\to\pm\infty$. However, since $\rho_\TF$ is integrable, we must have $\Phi_\TF(+\infty) = \Phi_\TF(-\infty) = \lambda$. Therefore, the total dipolar moment is null:
\[
    0 = \Phi_\TF({+}\infty) - \Phi_\TF( - \infty) = 4 \pi \int_\R y \left(\rho - \mu\right)(y) \rd y.
\]
This proves Proposition~\ref{prop:properties_PhiTF}.


\subsubsection{Sommerfeld estimates when $\mu$ is compactly supported. Proof of Proposition~\ref{prop:sommerfeld}}
\label{ssec:Sommerfeld}
We now consider the special case where $\mu$ is compactly supported, say in the interval $[a, b]$, and prove the Sommerfeld estimates in Proposition~\ref{prop:sommerfeld}. Outside of $[a,b]$, $\Phi_\TF$ satisfies
\begin{equation} \label{eq:def:TFeqt_outside_mu}
\Phi_\TF''(x) = - 4 \pi \left( \frac{3}{5 c_\TF} \left[ \lambda - \Phi_\TF(x)  \right] \right)^{3/2}.
\end{equation}
We solve this ordinary differential equation explicitly. First, on the interval $(b, +\ii)$, we multiply~\eqref{eq:def:TFeqt_outside_mu} by $\Phi_\TF'$ and integrate to obtain that
\[
    \frac12 \left| \Phi_{\TF}' \right|^2(x) = \left( \frac{3}{5 c_\TF} \right)^{3/2} \frac{8 \pi}{5} \left[\lambda  - \Phi_\TF(x)\right]^{5/2} + cst.
\]
As $x \to \infty$, we have $\Phi_\TF(x) \to \lambda$ and $\Phi_\TF'(x) \to 0$, so the integration constant is null. In addition, since $\Phi''_\TF=-4\pi \rho \le 0$, $\Phi_\TF'$ is decreasing, and goes to $0$ at infinity, hence $\Phi_\TF' \ge 0$ on $[b, \infty)$. Taking square roots gives
\[
    \dfrac{ \Phi'_\TF(x)}{ \left[ \lambda - \Phi_\TF(x)  \right]^{5/4}} = \frac{4 \sqrt{\pi}}{\sqrt{5}}  \left( \frac{3}{5 c_\TF} \right)^{3/4}.
\]
Integrating a second time shows that there is $x_b \in \R$ so that
\[
    \dfrac{1}{\left[\lambda - \Phi_\TF(x) \right]^{1/4}} =  \frac{ \sqrt{\pi}}{\sqrt{5}}  \left( \frac{3}{5 c_\TF} \right)^{3/4} (x  - x_b).
\]
So, for all $x \ge b$, we have
\[
    \Phi_\TF(x) = \lambda - \dfrac{c_1}{ ( x - x_b )^4} , \quad \text{with} \quad
    c_1 := \dfrac{5^5 c_{\TF}^3}{3^3 \pi^2}.
\]
In addition, since $\Phi_\TF'' = - 4 \pi \rho$, we obtain that, for all $x \ge b$, we have
\begin{equation} \label{eq:solution_rhoTF}
    \rho(x) = \dfrac{c_2}{(x - x_b)^6}, \quad \text{with} \quad
    c_2 = \frac{5c_1}{\pi}.
\end{equation}
We have similar results on $(-\ii,a)$. In particular, we obtain
\[
    \lim_{x \to \infty} \left[ \lambda - \Phi_\TF(x)\right] | x |^4 = -c_1, 
    \quad \text{and} \quad
    \lim_{x \to \infty} | x |^6 \rho(x) = c_2.
\]
As we already mentioned, these limits are independent of the system under consideration. This concludes the proof of Proposition~\ref{prop:sommerfeld}.

\section{Homogeneous 2-d  materials in the reduced Hartree-Fock model}
\label{sec:rHF}
We now prove our results concerning the rHF model. 


\subsection{Trace and kinetic energy per unit-surface}\label{ssec:trv}
In this subsection, we define both the trace per unit surface $\VTr$ and the kinetic energy per unit surface. For $\bR \in \R^2$, we denote by $\tau_\bR$ the translation operator on $L^2(\R^{3})$ given by $(\tau_\bR f)(\bx,x_3) = f(\bx - \bR, x_3)$.
Another way to write ~\eqref{eq:assumption_gamma} for a density matrix $\gamma$ is
\begin{equation} \label{eq:assumption_gamma_bis}
	\tau_\bR \gamma = \gamma \tau_\bR,\qquad \text{for all}\; \bR \in \R^2. 
\end{equation}
For $\gamma$ satisfying this condition, we define the trace per unit surface
\[
\VTr_3(\gamma) := \Tr_3 \left( \1_\Gamma \gamma \1_\Gamma  \right),
\]
where $\Gamma$ is the tube $[-\frac12, \frac12]^2 \times \R$. If, in addition, $\gamma$ is locally trace-class, with density $\rho_\gamma$, then $\rho(x_1, x_2, x_3) = \rho(x_3)$ depends only on the third variable and
\[
\VTr_3(\gamma) = \int_{\R} \rho_\gamma(x_3) \rd x_3.
\]
The space of admissible states $\cP$ is defined by
\[
\cP :=  \left\{ \gamma \in \cS(L^2(\R^3)) : \; \ 0 \le \gamma \le 1, \; \VTr(\gamma) < \infty \,\text{and} \,  \tau_\bR \gamma = \gamma \tau_\bR \; \text{for all}\; \bR \in \R^2 \right\}.
\]

Since the elements of $\cP$ commute with all $\R^2$--translations, we can apply Bloch-Floquet theory~\cite[Section XIII--16]{ReedSimon4} (see also \cite{Lingling1}). Let $\cF$ be the partial Fourier transform defined on $C^\infty_0(\R^{3})$ by
\[
\left( \cF f \right)(\bk, z) := \frac{1}{2 \pi} \int_{\R^2} \re^{- \ri \bk \cdot \by} f(\by,z) \rd \by,
\]
and extended by density to $L^2(\R^3)$. The map $\cF$ is unitary on $L^2(\R^3)$. Since $\gamma \in \cP$ commutes with $\R^2$-translations, we have
\begin{equation} \label{eq:def:gammak}
	\cF \gamma \cF^{-1} = \int_{\R^2}^\oplus \gamma_\bk \rd \bk,
\end{equation}
that is, for all $ f \in L^2(\R^3)$,
\[ 
\left( \cF \gamma f \right)(\bk, \cdot) =  \gamma_\bk \left[ \left(\cF f \right) (\bk, \cdot )\right] .
\]
Here, $(\gamma_\bk)_{\bk \in \R^2}$ is a family of self-adjoint operators acting on $L^2(\R)$. In terms of kernels, we formally have
\begin{equation} \label{eq:gamma-ker=gamma_k-ker}
	\gamma(\bx,x_3 ; \by,y_3)=\frac{1}{(2\pi)^2}\int_{\R^2}\re^{-\ri \bk \cdot (\bx-\by)}\gamma_\bk(x_3,y_3) \rd \bk.
\end{equation}
In particular,
\begin{equation} \label{eq:gamma_gammak}
	\rho_\gamma(x_3)=\frac{1}{(2\pi)^2}\int_{\RR^2}\rho_{\gamma_\bk}(x_3) \rd \bk,
	\quad \text{and} \quad
	\VTr_3 \left( \gamma \right) = \dfrac{1}{(2 \pi)^2} \int_{\R^2} \Tr_1 \left(\gamma_\bk\right) \rd \bk.
\end{equation}
Finally, we have
$$ \cF \Delta_3 \cF^{-1} = | \bk |^2 + \Delta_1,$$ 
thus, the kinetic energy per unit surface of $\gamma$ is given by
\begin{equation} \label{eq:def:kineticEnergy_perUnitSurface}
	\frac12 \VTr_3( - \Delta_3 \gamma) = \frac12 \dfrac{1}{(2 \pi)^2} \int_{\R^2} \left( | \bk |^2 \Tr_1(\gamma_\bk) + \Tr_1( - \Delta_1 \gamma_\bk ) \right) \rd \bk.
\end{equation}

\subsection{Reduced states: Proof of Theorem~\ref{th:gamma_G}}\label{ssec:reduction-rHF}
 For $\gamma \in \cP$, we associate the reduced operator $G_{\gamma} \in \cS(L^2(\R))$ defined by (compare with~\eqref{eq:def:gammak}: the superscript $\oplus$ is no longer here)
\begin{equation} \label{eq:def:G_gamma}
	\forall f \in L^2(\R), \quad (G_\gamma f) := \dfrac{1}{(2 \pi)^2} \int_{\R^2} (\gamma_\bk f) \rd \bk,
\end{equation}
The kernel of $G_{\gamma}$ is given by (compare with~\eqref{eq:gamma-ker=gamma_k-ker})
\[
    G_{\gamma}(x_3, y_3) = \frac{1}{(2\pi)^2}\int_{\R^2} \gamma_\bk(x_3,y_3) \rd \bk.
\]

Since $\gamma$ is a positive operator, so are its fibers $\gamma_\bk$, hence $G_{\gamma} \ge 0$ as well. In addition, we have
\[
\rho_{G_{{\gamma}}} = \dfrac{1}{(2 \pi)^2} \int_{\R^2} \rho_{\gamma_\bk} \rd \bk = \rho_\gamma, 
\quad \text{and} \quad 
\Tr_1(G_\gamma) =  \dfrac{1}{(2 \pi)^2} \int_{\R^2} \Tr_1( \gamma_\bk ) \rd \bk = \VTr_3(\gamma).
\]
This proves that $G_{{\gamma}} \in \cG = \left\{ G \in \cS(L^2(\R)), \ G \ge 0, \ \Tr_1(G) < \infty \right\}$.

We now prove the inequality~\eqref{eq:ineq:kineticGamma_G}.  According to~\eqref{eq:def:kineticEnergy_perUnitSurface} we have
\[
\frac12 \VTr_3( - \Delta_3 \gamma) = \frac12 \dfrac{1}{(2 \pi)^2} \int_{\R^2} | \bk |^2 \Tr_1(\gamma_\bk) \rd \bk + \frac12 \Tr_1(-\Delta_1 G_\gamma).
\]
Unfortunately, the first term cannot be expressed directly in terms of the operator $G_\gamma$. We only have an inequality for this term. Since $G_\gamma$ is trace-class, it is compact, and it has a spectral decomposition of the form
\[
G_\gamma = \sum_{j = 1}^\infty g_j | \phi_j \rangle \langle \phi_j |, 
\]
where $(\phi_j)_j$ is an orthonormal basis of $L^2(\R)$, composed of eigenvectors of $G_\gamma$ with $g_j \ge 0$ and $\sum g_j < \infty$. We denote by 
\[
m_j(\bk) := \langle \phi_j, \gamma_\bk \phi_j \rangle.
\]
Since $0 \le \gamma_\bk \le 1$, we have  $0 \le m_j(\bk) \le 1$. In addition, from~\eqref{eq:def:G_gamma}, it follows that
\[
    \forall j \in \N, \quad g_j = \dfrac{1}{(2 \pi)^2} \int_{\R^2} m_j(\bk) \rd \bk.
\]
We deduce that
\begin{align*}
	& \dfrac{1}{2 (2 \pi)^2} \int_{\R^2} | \bk |^2 \Tr_1(\gamma_\bk) \rd \bk  =
	\sum_{j=1}^\infty  \dfrac{1}{2 (2 \pi)^2} \int_{\R^2} | \bk |^2 m_j(\bk) \rd \bk  \\
	& \qquad \ge 
	\sum_{j=1}^\infty \inf \left\{   \dfrac{1}{2 (2 \pi)^2} \int_{\R^2} | \bk |^2 m(\bk) \rd \bk, \quad 0 \le m(\bk) \le 1, \  \dfrac{1}{(2 \pi)^2} \int_{\R^2} m(\bk) \rd \bk = g_j  \right\}.
\end{align*}
According to the bathtube principle (see~\cite[Thm 1.14]{lieb2001analysis}), the last minimization problem admits a unique minimizer, of the form $\1 \left(| \bk | \le k_F \right)$. The value of the radius is found with the condition $\int_{\R^2} m = (2 \pi)^2 g_j$ and we get
\begin{equation} \label{eq:def:mjstar}
	m_j^*(\bk) := \1 \left( | \bk | \le 2 \sqrt{\pi g_j}   \right).
\end{equation}

For this value, we have
\[
\sum_{j=1}^\infty \dfrac{1}{2 (2 \pi)^2} \int_{\R^2} | \bk |^2 m_j^*(\bk) \rd \bk = \sum_{j=1}^\infty \pi g_j^2 = \pi \Tr_1 (G_\gamma^2),
\]
which proves~\eqref{eq:ineq:kineticGamma_G}. 

\medskip
Conversely, for any $G \in \cG$, if we write $G =\sum_{j=1}^\infty g_j | \phi_j \rangle \langle \phi_j |$ and set 
\begin{equation} \label{eq:def:opt_gamma}
	\cF \gamma^* \cF^{-1} := \int_{\R^2}^\oplus \gamma_\bk \rd \bk, 
	\quad \text{with} \quad
	\gamma_\bk := \sum_{j=1}^\infty m_j^*(\bk) | \phi_j \rangle \langle \phi_j |,
\end{equation}
we see that $G_{\gamma^*} = G$ (representability) and~\eqref{eq:ineq:kineticGamma_G} becomes an equality. This proves Theorem~\ref{th:gamma_G}.

\begin{remark} \label{rem:higher_dim}
	The proof in higher dimension $\R^{d+s}$ (see Remark~\ref{rem:general_case_inequality}) is similar. Indeed, set
	\[
	G = \dfrac{1}{(2 \pi)^s} \int_{\R^s} \gamma_\bk \rd \bk=\sum_{j=1}^\infty g_j | \phi_j \rangle \langle \phi_j |,
	\]
	and, for every $j\in\N$, the optimal $m_j^*$ is defined by 
	\[
	m_j^*(\bk) := \1 \left(| \bk | \le c_s g_j^{1/s}\right) 
	\quad \text{with} \quad
	c_s := 2 \pi \left( \frac{s}{\left| \SS^{s-1} \right|} \right)^{1/s}.
	\]
	This yields
	\[
	\sum_{j=1}^\infty \dfrac{1}{2 (2 \pi)^s} \int_{\R^s} | \bk |^2 m_j^*(\bk) \rd \bk = c_\TF(s) \Tr \left( G^{1 + \frac{2}{s}} \right),      
	\]
	where $c_\TF(s) := \frac{s}{s+2} \left(\frac{s}{\left| \SS^{s-1} \right|}\right)^{2/s} 2 \pi^2$ is the $s$-dimensional (spinless) Thomas--Fermi constant. 
\end{remark}

Thanks to the reduction of the kinetic energy, and reasoning as in the Thomas--Fermi section, we obtain that $\cI^\rHF_{3} = \cI^\rHF$.
This proves Theorem~\ref{th:rHF_equivalence}.


\subsection{Existence of minimizers for the reduced model: Proof of Theorem~\ref{th:main_rHF}}
\label{ssec:rHF_is_well_posed}

We now focus on the reduced rHF problem. 

\subsubsection{Existence and uniqueness of the minimizer}
Let us first prove that the minimization problem $\cI^\rHF$ is well-posed. We start by noting that $\cG$ is not empty. Indeed, for $\rho=\frac12\mu*e^{-\av{x}}$, we have $\av{ \sqrt{\rho}\rangle \langle \sqrt{\rho}}\in\cG$.  Let $(G_n) \subset \cG$ be a minimizing sequence satisfying $\Tr_1(G_n) = Z$. Then $\cE^\rHF(G_n)$ is bounded, and since it is the sum of three positive terms, there is $C \ge 0$ so that
\[
\Tr_1( - \Delta G_n) \le C, \quad \Tr_1 (G_n^2) \le C, \quad \text{and} \quad \cD_1(\rho_n - \mu) \le C,
\]
where we set $\rho_n := \rho_{G_n}$. In addition, $\norm{\rho_n}_{L^1}=Z$, 
thus $(\rho_n)$ is bounded in $L^1(\R)$.   
We then deduce that, up to a subsequence, still denoted by $(G_n)$ and $(\rho_n)$, we have the following weak-* convergences (we denote by $\fS_p := \fS_p(L^2(\R))$ the $p$ Schatten class with $\| G \|_{\fS_p}^p = \Tr_1(| G |^p)$)
\begin{align*}
	| \nabla | G_n | \nabla |  \hookrightarrow T \quad &\text{weakly-* in $\fS_1$}  \\
	G_n  \hookrightarrow G_* \quad & \text{weakly-* in $\fS_2 \cap \fS_1$} \\
	\rho_n  \hookrightarrow \rho_* \quad & \text{weakly in} \; L^1(\R),\\
	W_{\rho_n -\mu}  \hookrightarrow W \quad & \text{weakly in}\; L^2(\R)
\end{align*}
where $W_{\rho_n -\mu}$ is defined in Section~\ref{ssec:1dHartree}. By standard arguments, we have that $T = | \nabla | G_* | \nabla |$, that $\rho_* = \rho_{G^*}$, and that $W=W_{\rho_* -\mu}$. In particular, the last equality shows that $W_{\rho_* -\mu} \in L^2(\R)$. In particular, we have $(\rho_* - \mu) \in \cC$, which implies the neutrality $\int_\R (\rho_* - \mu) = 0$. Hence $\Tr(G)=Z$. Furthermore, we have $\cE^\rHF(G_*) \le \displaystyle\liminf_n \cE^\rHF(G_n)$, thus $G^*$ is a minimizer of $\cE^\rHF$.

This minimizer is unique, thanks to the {\em strict} convexity of $\cE^\rHF$ due to the $\Tr_1(G^2)$ term.


\subsubsection{Derivation of the Euler-Lagrange equations}
We now derive the Euler-Lagrange equations. In what follows, we denote by $G_*$ the unique minimizer of the $\cE^\rHF$ functional. For all $G \ge 0$ with $\Tr_1(G) = Z$, and all $0 \le t \le 1$, one has 
\[
\cE^\rHF( (1 - t) G_* + t G) \ge \cE^\rHF(G_*).
\]
This gives 
\begin{equation}\label{eq:cd-optimalite}
\Tr_1 \left( \left[ - \tfrac12 \Delta + 2 \pi  G_* + \Phi_* \right] (G - G_*)   \right) \ge 0,
\end{equation}
where $\Phi_*$ is the mean-field potential generated by $\rho_{G_*} - \mu$ (see Proposition~\ref{prop:1D-hartree}). Let us denote by 
\[
h := H_* +  2 \pi G_*, \quad \text{where} \quad 
H_* := - \frac12 \Delta + \Phi_*.
\]
Testing~\eqref{eq:cd-optimalite} over states of the form $G=Z\av{\phi\rangle  \langle \phi}$, with $\| \phi \|_{L^2} = 1$ shows that $h$ is bounded from below. Let $\lambda := \inf \sigma \left(h \right) > - \infty$. Since $\Tr_1(G) = Z = \Tr_1(G_*)$, for all $G\in\cG$, the inequality~\eqref{eq:cd-optimalite} can also be written as 
\[
\Tr_1((h - \lambda) G) \ge \Tr_1((h - \lambda)G_*),\qquad \forall G\in\cG.
\]
In particular, the minimization problem
$$
\inf\left\{ \Tr_1 \left( (h - \lambda) G \right), \; G \in\cG, \; \Tr_1(G) = Z   \right\}
$$
is well-posed and admits a unique minimizer, which is $G_*$. By definition of $\lambda$, the  above minimum is $0$. It follows that $\Tr_1((h-\lambda)G_*) = 0$, hence ${\rm Ran} \, G_* \subset {\rm Ker}(h- \lambda)$. In particular, $ {\rm Ker}(h- \lambda) \neq \{ 0 \}$, so that $\lambda$ is an eigenvalue of $h$. Moreover, we have
\[
(h-\lambda ) G_* = G_* (h - \lambda) = 0.
\]
Let us consider the spectral decomposition of $G_*$, of the form $G_* = \sum_{j=1}^\infty g_j | \phi_j  \rangle \langle \phi_j |$. Then, for all $j \geq 1$ with $g_j > 0$, we have $h \phi_j = \lambda \phi_j$, that is
\[
H_* \phi_j +  2 \pi g_j \phi_j = \lambda \phi_j.
\]
So $\phi_j$ is an eigenvector of $H_*$ with corresponding eigenvalue
\[
\varepsilon_j := \lambda - 2 \pi g_j. 
\]
As $g_j > 0$, then $g_j = \dfrac{1}{2 \pi} \left( \lambda -\varepsilon_j \right)_+$. This proves that $G_*$ is of the form
\begin{equation} \label{eq:decomposition_Gstar}
	G_* =  \sum_j \dfrac{1}{2 \pi} \left( \lambda -\varepsilon_j \right)_+ | \phi_j \rangle \langle \phi_j |,
	\quad \text{with} \quad
	H_* \phi_j = \varepsilon_j \phi_j.
\end{equation}
Conversely, if $\varepsilon < \lambda$ is an eigenvalue of $H_*$ with corresponding eigenvector $\psi$, one has
\[
0 = G_*(h - \lambda) \psi = G_* ( H_* - \lambda +  2 \pi  G_*) \psi
= 2 \pi G_* \left( G_* \psi \right) + (\varepsilon - \lambda) G_* \psi .
\]
Denoting by $\phi := G_* \psi$ and $g:= (\lambda - \varepsilon )/2 \pi > 0$, we have $G_*\phi=g\phi$. We claim that $\phi\neq 0$. Otherwise, we would have $h \psi = H_* \psi = \varepsilon \psi$, and $\varepsilon$ would be an eigenvalue of $h$, smaller that $\lambda < \inf \sigma(h)$, a contradiction. So $\phi$ is an eigenvector of $G_*$ with eigenvalue $g$.  According to the previous decomposition of $G_*$, $\phi$ corresponds to one of the $\phi_j$, up to a multiplicative factor. In other words, all eigenvalues of $H_*$ smaller than $\lambda$ are considered in the decomposition~\eqref{eq:decomposition_Gstar}. Besides, as $G_*$ is a compact operator, $H_*$ is a compact perturbation of $h$. Therefore $\sigma_{\rm ess}(H_*)=\sigma_{\rm ess}(h)$, and no essential spectrum of $H_*$ can lie below $\lambda$. We deduce that
\[
G_* = \dfrac{1}{2 \pi}  \left( \lambda - H_*\right)_+.
\]
This ends the proof of Theorem~\ref{th:main_rHF}. 

\subsection{Properties of the mean-field potential: Proof of Proposition~\ref{prop:Phi_rHF}} 
\label{ssec:props-V-rHF}

We finally prove some properties of the mean-field potential $\Phi_*$ in the case where $| x | \mu(x)  \in L^1(\R)$. First, we note that, for $x \ge 0$, 
\[
    \Phi_*(x)  = 2 \pi \int_{\R^+} (\rho - \mu)(y) \min \{ x, y \} \rd y \ge - 2 \pi \int_{\R^+} \mu(y) y \rd y > - \infty,
\]
and similarly for $x \le 0$, so $\Phi_*$ is bounded from below. Let us prove that $\Phi_*$ is also bounded from above. This will imply, as in the Thomas Fermi case that $| x | \rho(x) \in L^1(\R)$ and that (see Section~\ref{ssec:screening_TF})
\[
    \lim_{x \to +\infty} \Phi_*(x) = 4 \pi  \int_{\R^+} y  (\rho_* - \mu)(y) \rd y 
    \quad \text{and} \quad
    \lim_{x \to -\infty} \Phi_*(x) = - 4 \pi  \int_{\R^-}  y  (\rho_* - \mu)(y) \rd y.
\]
Assume that $\lim_{x \to +\infty} \Phi(x) = +\infty$. Then, we would also have $\int_0^\infty \rho_*(y) y \rd y= +\infty$. The contradiction comes from Agmon's estimates, which state that if the potential $\Phi(x)$ goes to $+\infty$, then the corresponding eigenvectors (hence the density $\rho$) are exponentially decaying. We provide a simple proof of Agmon's argument in the one-dimensional setting for completeness. This is not the optimal result, and we refer to the original work~\cite{Agmon} for details (see for instance the example at the end of Chapter 1 in~\cite{Agmon}).

\begin{lemma}\label{lem:agmon}
	Let $V \in L^1_{\rm loc}(\R)$ be a potential  bounded from below and let $E_0 \in \R$.  Assume that $\lim_{x \to +\infty} V(x) = +\infty$ and let $a \in \R$ be such $V(x) > E_0 + 1$ for all $x > a$. Then, there is $C > 0$ such that, for all $u \in H^1(\R)$ eigenvector of the operator $H = - \partial_{xx}^2  + V  $ associated with an eigenvector $E \le E_0$, we have
	\[
	\int_a^\infty \re^{(x - a)} | u  |^2(x) \rd x \le C.
	\]
\end{lemma}

\begin{proof}
	Let $\psi \in C^\infty_0(\R)$ and consider the test function $\phi := \psi^2 u$. The equation $\langle \phi, (H - E) u \rangle = 0$ becomes
	\[
	\int_{\R} u' (\psi^2 u)' + \int_{\R} (V - E) \psi^2 u^2  = 0.
	\]
	Together with the identity $u'  (\psi^2 u)' = | (\psi u)' |^2 - | \psi' |^2 u^2$, this gives
	\[
	\int_\R | \psi' |^2 u^2  =  \int_{\R} | (\psi u )' |^2 + \int_{\R} (V - E) \psi^2 u^2 \ge  \int_{\R} (V - E) \psi^2 u^2.
	\]
	Assume that $\psi$ has support in $(a, \infty)$. On this support, we have $(V - E) \ge 1$, so
	\begin{equation}\label{eq. local ineq}
		\int_\R \psi^2 u^2 \le \int_\R | \psi ' |^2 u^2.
	\end{equation}
	By density, this inequality remains valid for all $\psi \in H^1_0((a, \infty))$. We choose $\psi$ of the form
	\[
	\psi(x) 
	:= \begin{cases}
		\re^{\frac12(x-a) } - 1 & \quad \text{for} \ x \in [a, L) \\ 
		\eta(x):=\left[ \left( \re^{\frac12(L-a) } - 1\right) - \varepsilon (x - L) \right]_+ & \quad \text{on} \ [L, +\ii),
	\end{cases}
	\]
	which is continuous and slowly decaying to $0$ (at rate $\varepsilon$) on $[L, \infty)$. Applying \eqref{eq. local ineq}, we obtain
	\begin{align*}
		\int_{a}^{L}(\re^{\frac12(x-a)}-1)^2 u^2(x) \rd x + \int_{L}^{\infty}\eta^2(x) u^2(x) \rd x	\le \frac{1}{4}\int_{a}^{L} \re^{x-a} u^2(x) \rd x + \int_{L}^{\infty}\eta'(x)^2 u^2(x) \rd x.
	\end{align*}
     Taking $\varepsilon \to 0$, so that $\eta'(x) \to 0$ in $L^\infty(\R)$, we obtain
	\begin{align}\label{eq. loc eq2}
		\int_{a}^{L} \left( \frac{3}{4}\re^{(x-a)} -2\re^{\frac12(x-a)} +1 \right) u^2(x) \rd x \le 0.
	\end{align}   
	On the other hand, one can find $C>0$ such that
	\[  \frac{3}{4}\re^{(x-a)} -2\re^{\frac12(x-a)} +1 \ge \frac12(\re^{(x-a)} -C).\]
	So that, \eqref{eq. loc eq2} becomes
	\[ \int_{a}^{L}  \re^{(x-a)} u^2(x) \rd x \le C \int_{a}^{L} u^2(x) \rd x\le C. \]
	Letting $L \to \infty$ proves the result.
\end{proof}

In our case, we consider all eigenvalues below $E_0 = \lambda$. This gives
\[
\int_a^\infty \rho(x) \re^{ (x - a )} \rd x = \sum_{i} (\lambda - \varepsilon_i)_+ \int_a^\infty \re^{ (x - a )} | u_i |^2(x) \rd x \le C \sum_i (\lambda - \varepsilon_i)_+ = CZ.
\]
This contradicts the fact that $\int_0^\infty y \rho(y) = \infty$.

\subsection{Properties of the density: Proof of Proposition~\ref{prop:rhoG}}
\label{ssec:proof_rhoG}

\review{ 
In this section, we make the stronger assumption that $| x |^3 \mu(x) \in L^1(\R)$. Our goal is to prove that $G_*$ is finite rank. }

\review{Before we prove this point, we make some remarks. First, since $\Phi_*$ is continuous and have some well-defined limits at $\pm \infty$, we have $\Sigma_{\rm ess} := \min \sigma_{\rm ess}(H_*) =  \min \{ \Phi_*(-\infty), \Phi_*(+\infty) \}$. Recall that $\lambda \le \Sigma_{\rm ess}$. The eigenvalues of $H_*$ can only accumulate at $\Sigma_{\rm ess}$, so if $\lambda <  \Sigma_{\rm ess}$, then $G_*$ is automatically finite-rank. On the other hand, if $G_*$ is finite rank, it is of the form $G_* = (2 \pi)^{-1} \sum_{j=1}^J (\lambda - \varepsilon_j) | \phi_j \rangle \langle \phi_j |$ (see \eqref{eq:decomposition_Gstar}) with $\varepsilon_j < \lambda$ for all $1 \le j \le J$. Then, by Agmon estimates, all eigenvectors of $G_*$ are exponentially decaying, hence so is $\rho_*$.}

\review{Unfortunately, we were not able to fully prove that $G_*$ is always finite rank. We can prove this fact whenever $| x |^3 \rho_*(x) \in L^1(\R)$ as well. In other words, if $\rho_*$ is not exponentially decaying, it has to be slowly decaying (no finite moment of order $3$). Although we expect that we always have $| x |^3 \rho_*(x) \in L^1(\R)$, (this is somehow confirmed by the numerical simulations), we could not find an argument for it.}

\review{Assume $| x |^3 \rho_*(x) \in L^1(\R)$, so that $f := \rho_* - \mu$ satisfy $| x |^3 | f |(x) \in L^1(\R)$ and $\int_\R f = 0$. We want to bound the number of negative eigenvalues of $H_* - \lambda = - \Delta + V$, with $V(x) := \Phi_*(x) - \lambda$. First, by the min-max theorem, it is enough to bound the number of negative eigenvalues of $- \Delta - V_-$, with $V_- := \max \{ 0, - V\}$. Using the first expression in~\eqref{eq:expression_Phi}, and the fact that $\lambda \le \Phi_*(+ \infty)$, we have, for $x > 0$, that
\[
     V(x)_- \le  \left| \Phi_*(+ \infty) - \Phi_*(x) \right| \le 4 \pi \int_0^\infty \left| t - \min \{ x, t \} \right| \cdot | f |(t) \rd t = 4 \pi \int_x^\infty (t - x) | f | (t) \rd t.
\]
Together with Fubini, this gives
\[
    \int_{\R^+} | x | V(x)_- \rd x \le 4 \pi \iint_{0 < x < t} x (t - x) | f |(t) \rd t \rd x = \frac{2 \pi}{3} \int_0^\infty t^3 | f |(t) \rd t < \infty.
\]
One has a similar estimate for the integral over $\R^-$, and we deduce that $\int_\R | x | V(x)_- \rd x < \infty$. One can now apply a Bargmann type bound~\cite[eq. 27]{chadan}, which states that $-\Delta - V_-$ has a finite number of negative eigenvalues. This concludes the proof. }

\section{Numerical illustrations}
\label{sec:numerical}

In this section, we provide numerical results for our reduced one-dimensional models. 

\subsection{Numerical setting and self-consistant procedure}
We use a simple code with finite differences: all functions are evaluated on a fine grid representing some interval $[-a, a]$, with $N_b$ points (we took $a = 15$ and $N_b = 5001$). The operator $-\Delta_1$ is computed in Fourier space (using sparse matrices). Given a neutral function $f$, we evaluate the potential $\Phi_f$ using~\eqref{eq:def:Vf}.

To solve the Thomas--Fermi problem, we use the following iterations. Recall that the optimal density $\rho_\TF$ satisfies the Euler-Lagrange equation~\eqref{eq:TF-EL}, that we write in the form
\[
    \rho_\TF  = \left( \frac{3}{5 c_\TF} \left( \Phi_\TF - \lambda \right)_+  \right)^{3/2}.
\]
At the point $\rho_n$, we set $\Phi_n = \Phi_{\rho_n - \mu}$, and compute $\lambda_n \in \R$ so that
\[
    Z_n(\lambda_n) = Z, \quad \text{where} \quad Z_n(\lambda ) := \int_{\R} \left(  \frac{3}{5 c_\TF} \left( \Phi_n - \lambda \right)_+  \right)^{3/2}.
\]
Since the map $Z_n(\cdot)$ is decreasing, one can efficiently compute $\lambda_n$ using a simple dichotomy. We then set
\[
    \widetilde{\rho_{n+1}}  :=  \left( \frac{3}{5 c_\TF} \left( \Phi_n - \lambda_n \right)_+  \right)^{3/2}, \quad \text{and} \quad
    \rho_{n+1} := t  \widetilde{\rho_{n+1}} + (1 - t) \rho_n,
\]
where $t$ is optimized to lower the Thomas--Fermi energy on the segment $[0, 1]$ (linear search).

To solve the rHF problem, we use a similar iterative procedure. Recall that $G_*$ satisfies the Euler-Lagrange equations
\[
    G_* = \frac{1}{2 \pi} \left( \lambda - H_* \right)_+.
\]
At point $G_n \in \cG$, we set $H_n := - \frac12 \Delta_1 + \Phi_n$ and find $\lambda_n$ so that 
\[
        \Tr_1 \left(  \frac{1}{2 \pi} \left( \lambda_n - H_n \right)_+  \right)= Z.
\]
Again, this can be solved using a dichotomy method. We then set
\[
    \widetilde{G_{n+1}}  := \frac{1}{2 \pi} \left( \lambda_n - H_n \right)_+, \quad \text{and} \quad
    G_{n+1} := t  \widetilde{G_{n+1}} + (1 - t) G_n,
\]
For the optimization problem in $t \in [0, 1]$, we note that the map $t \mapsto \cE^\rHF( t \widetilde{G_{n+1}} + (1 - t) G_n)$ is quadratic in $t$, so the best $t \in [0, 1]$ is explicit.

Although we believe that more sophisticated methods can be designed to study these problems, our numerical methods for both the TF and rHF problems seem to converge quite fast, which is enough for our purpose.

\subsection{Numerical results}

We now provide some numerical results in three test cases. In order to compare the (spinless) rHF and TF results, we use the spinless $c_\TF$ constant $c_\TF := \frac{3^{5/3} \pi^{4/3} }{2^{1/3} 5}$ for the TF model (see Remark~\ref{rem:higher_dim}).

\subsubsection{Case 1: a simple slab}
For our first test case, we consider the charge 
\[
    \mu_1(x) := \1(| x | < 2),
\]
which models an homogeneous charged slab having some width. The results are displayed in Figure~\ref{fig:mu1}. In the (A) part, we display the function $\mu_1$, together with the best rHF density $\rho_*$ and the generated potential $\Phi_*$. The potential is scaled by a factor $10$ for ease of reading. We do not display the TF results here, as they are very similar to the rHF ones. In the (B) part, we display the difference between the TF results and the rHF ones, and we plot $\rho_\TF - \rho_*$ and $\Phi_\TF - \Phi_*$. Although these two densities are not equal, they are very close, and they only differ by around $1\%$. In addition, this difference seems to be even lower far from the slab ($| x |$ large). The potentials $\Phi_\TF$ and $\Phi_*$ differ by a small constant as they are defined as the unique function solving the Poisson equation $\Phi''=-4\pi (\rho-\mu)$ with the conditions $\Phi'\to 0$ at $\pm\infty$ and $\Phi(0)=0$. 

The optimal operator $G_*$ that we have numerically found has 15 positive eigenvalues, 4 of which being smaller than $10^{-9}$. The remaining 11 other eigenvalues are greater than $10^{-3}$, and the larger is around $1.07$ (which shows that $G$ can have eigenvalues greater than $1$).

\begin{figure}[H]
    \centering
    \begin{subfigure}[b]{0.45\textwidth}
        \includegraphics[width=\textwidth]{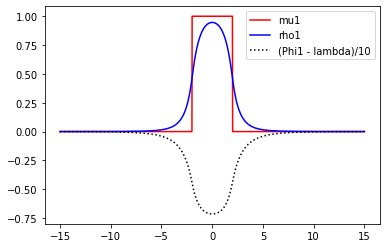}
        \caption{$\mu_1$ (red), $\rho_*$ (blue) and $(\Phi_* - \lambda_*)/10$ (dotted black).}
    \end{subfigure}
    \begin{subfigure}[b]{0.45\textwidth}
        \includegraphics[width=\textwidth]{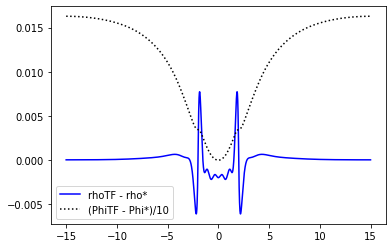}
        \caption{The difference $\rho_{\TF} - \rho_*$ (blue) and $\Phi_\TF - \Phi_*$ (dotted black).}
    \end{subfigure}
\caption{Results for the slab $\mu_1$.}
\label{fig:mu1}
\end{figure}

\subsubsection{Case 2: two slabs}
For our second test case, we consider two different slabs. We take the charge 
\[
    \mu_2(x) := \1(-5 < x < -2) + 2 \cdot \1(1 < x < 3).
\]
The main difference with the previous case is that $\mu_2$ has a non null dipolar moment. The results are displayed in Figure~\ref{fig:mu2}. Again, the TF results are very close to the rHF ones. This is surprising, since we expected the rHF model to exhibit some (screened) dipolar moment. However, we found numerically dipolar moment of order $2 \cdot 10^{-2}$ for the rHF case, and of order $1 \cdot 10^{-2}$ for the TF one. We believe that they come from numerical errors: recall that we are working in the finite box $[-15, 15]$, and that we expect $\Phi_\TF - \lambda$ to decay as $-c_1 | x |^{-4}$. 

We were not able to prove that the dipolar moment should vanish also in the rHF case. Still, even though it does not vanish, we believe that it can always be neglected.

The optimal operator $G_*$ in this case has rank $17$, with two eigenvalues of order $10^{-11}$, and the remaining ones of order $10^{-3}$ to $1$. The highest eigenvalue is around $1.44$.

\begin{figure}[H]
    \centering
    \begin{subfigure}[b]{0.45\textwidth}
        \includegraphics[width=\textwidth]{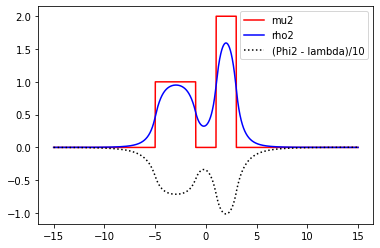}
        \caption{$\mu_2$ (red), $\rho_*$ (blue) and $(\Phi_* - \lambda_*)/10$ (dotted black).}
    \end{subfigure}
    \begin{subfigure}[b]{0.45\textwidth}
        \includegraphics[width=\textwidth]{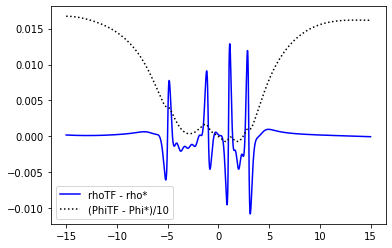}
        \caption{The difference $\rho_{\TF} - \rho_*$ (blue) and $\Phi_\TF - \Phi_*$ (dotted black).}
    \end{subfigure}
    \caption{Results for the step function $\mu_2$.}
    \label{fig:mu2}
\end{figure}

\subsubsection{Case 3: two slabs, smooth case}

Finally, we study the case where $\mu$ is smooth. We took
\[
    \mu_3(x) = \re^{ - \frac14 {(x+2)^2}} +  2 \cdot \re^{ - {(x-2)^2}}.
\]
The charge density models two slabs having non-homogeneous charge in the $x_3$-direction. The results are displayed in Figure~\ref{fig:mu3}. Again, the densities and mean-field potentials are very close. Actually, it seems that, due to smoothness, the difference $\rho_\TF - \rho_*$ is now of order $0.1\%$. In particular, it seems that the TF model is a very good approximation of the rHF one. 

The optimal operator $G_*$ in this case has rank $19$, with $4$ eigenvalues below $10^{-10}$, and the remaining ones above $10^{-3}$. The highest eigenvalue is $1.32$.

\begin{figure}[H]
    \centering
    \begin{subfigure}[b]{0.45\textwidth}
        \includegraphics[width=\textwidth]{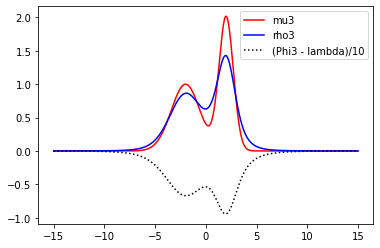}
        \caption{$\mu_3$ (red), $\rho_*$ (blue) and $(\Phi_* - \lambda_*)/10$ (dotted black).}
    \end{subfigure}
    \begin{subfigure}[b]{0.45\textwidth}
        \includegraphics[width=\textwidth]{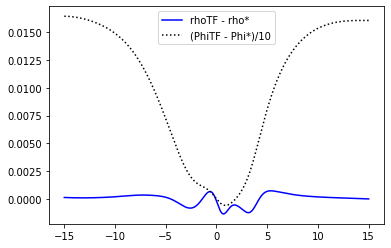}
        \caption{The difference $\rho_{\TF} - \rho_*$ (blue) and $\Phi_\TF - \Phi_*$ (dotted black).}
    \end{subfigure}
    \caption{Results for the step function $\mu_3$.}
    \label{fig:mu3}
\end{figure}

\subsection{Perspectives}

In view of these numerical results, we believe that, even for  more complex Kohn--Sham models, one can always approximate the reduced kinetic energy by the Thomas Fermi one. 

More specifically, starting from a three-dimensional Kohn--Sham (KS) model of the form
\[
    \cE^{\rm KS}_3(\gamma) := \frac12 \Tr_3(- \Delta \gamma) + \frac12 \cD_3(\rho_\gamma - \mu) + E_{3}^{\rm xc}(\rho_\gamma),
\] 
and assuming that the charge $\mu$ depends only on the third variable, one can assume that the optimal $\gamma$ will again satisfy~\eqref{eq:assumption_gamma} (although symmetry breaking can now happen due to the non-convexity of the models). If this is the case, then one can perform the same analysis as before, and obtain a reduced one-dimensional energy per unit surface, of the form
\[
    \cE^{\rm KS}(G) := \frac12 \Tr_1(G) + \pi \Tr_1(G^2) + \frac12 \cD_1(\rho_G - \mu) + E_1^{\rm xc}(\rho_G),
\]
where $E^{\rm xc}_1$ is the exchange correlation energy {\em per unit surface}. The corresponding Thomas--Fermi model is
\[
    \cE^{\rm KS, TF}(\rho) := c_\TF \int_{\R} \rho^{5/3}  + \frac12 \cD_1(\rho - \mu) + E_1^{\rm xc}(\rho).
\]
We believe that the optimal TF density $\rho_\TF$ is always very close to the optimal KS one $\rho_{\rm KS}$. The advantage is that $\cE^{\rm KS, TF}$ is easier to optimize numerically, and to study theoretically.

\appendix

\section{A Lieb-Thirring inequality}
\label{sec:LT}

In this section, we explain how to use Theorem~\ref{th:gamma_G} to obtain a Lieb-Thirring type inequality~\cite{LieThi-75, LieThi-76}. We state our result in the general dimension $d\in \N\setminus\set{0}$. 

\begin{proposition} \label{prop:LT}
    Let $d\in\N\setminus\set{0}$ and $G\in\cS(L^2(\R^d))$ be a positive operator. Then, for any $s\in\N\setminus\set{0}$,
    \begin{equation} \label{eq:LTtype}
    \boxed{
        K \left( \int_{\R^d} \rho_G^{1 + \frac{2}{d+s}}\right)^{1 + \frac{s}{d}} \le \left( \Tr_d( G^{1 + \frac{2}{s}}) \right)^{s/d} \Tr_d( - \Delta_d G)
    }
    \end{equation}
    with the constant
    \[
    K := \dfrac{ K_\LT(d+s)^{1 + \frac{s}{d}} }{ ( 2c_{\TF}(s))^{\frac{s}{d}} } \left( \dfrac{s}{d+s} \right)^{s/d} \dfrac{d}{d+s}.
    \]
    Here, $K_\LT(d+s)$ is the usual Lieb-Thirring constant in dimension $d+s$, that is the best constant in the inequality
    \begin{equation} \label{eq:usual_LT}
        \forall \gamma \in \cS\left( L^2(\R^{d+s})\right), \quad 0 \le \gamma \le 1, \quad K_\LT(d+s) \int_{\R^{d+s}} \rho_\gamma^{1 + \frac{2}{d+s}} \le \Tr_{d+s}( - \Delta \gamma).
    \end{equation}
\end{proposition}
\begin{proof}
Consider $G \in \cS(L^2(\R^d))$ , $G \ge 0$ such that $\Tr(G^{1 + \frac{2}{s}}) < \infty$ and consider the optimal $\gamma \in \cS(L^2(\R^{d+s}))$ as in~\eqref{eq:def:opt_gamma}. Then, the Lieb-Thirring inequality~\eqref{eq:usual_LT} applied to $\gamma$ gives (after the appropriate {\em per unit surface} normalization)
\[
\frac12 \Tr_d(- \Delta_d G) + c_{\TF}(s) \Tr_d(G^{1 + \frac{2}{s}}) = \frac12 \VTr_{d+s}(- \Delta_{d+s} \gamma) \ge \frac12 K_\LT(d+s) \int_{\R^d} \rho_\gamma^{1 + \frac{2}{d+s}},
\]
Since $\rho_\gamma = \rho_G$ then, for all $G \in \fS_{1 + \frac{2}{s}}(\R^d)$ such that $\Tr (- \Delta G) < \infty$, we get $\rho_G \in L^{1 + \frac{2}{d+s}}(\R^d)$, and
\[
\Tr_d(- \Delta_d G) + 2 c_{\TF}(s) \Tr_d(G^{1 + \frac{2}{s}}) \ge K_\LT(d+s) \int_{\R^d} \rho_G^{1 + \frac{2}{d+s}}.
\]
Performing the scaling $G_{\lambda} = \lambda G$ and optimizing over $\lambda$ gives the result.
\end{proof}
Proposition~\ref{prop:LT} corresponds to the Lieb-Thirring inequality for operators in $\cS(L^2(\R^{s+d}))$ in a semi-classical limit, when the semi-classical limit dilation is only performed in the first $s$ variables (see also~\cite{Martin1990new} for similar arguments).

This type of inequalities was recently studied in~\cite{Gontier2020}, where it is shown that for all $d \ge 0$ and $1 \le p \le 1 + \frac{2}{d}$, there is an optimal constant $K_{p,d}$ so that
\[
K_{p,d} \| \rho_G \|_p^{\frac{2p}{d(p-1)}} \le \| G \|_{\fS^{q}}^{\frac{p(2-d) + d}{d(p-1)}} \Tr_d (- \Delta_d G),
\quad \text{with} \quad
q := \dfrac{2p + d - dp}{2 + d - dp}.
\]
It is proved that this constant is the dual constant of the usual Lieb-Thirring constant $L_{\gamma, d}$ with $\gamma = q/(q-1)$. The case in Proposition~\ref{prop:LT} corresponds to the choice 
\[
p = 1 + \dfrac{2}{d+s} \quad \text{with} \quad s \in \N, \quad \text{so that} \quad q = 1 + \frac{2}{s}.
\]
This corresponds to the dual constant $L_{\gamma, d}$ with $\gamma = \frac{q}{q-1} = 1 + \frac{s}{2}$. In particular, since $s \ge 1$, we have $\gamma \ge \frac32$. In this regime, it is known that the best constant is the semi-classical one: $L_{\gamma, d} = L_{\gamma, d}^{\rm sc}$, hence $K_{p,d} = K_{p, d}^{\rm sc}$. This proves that the optimal constant $K$ in the inequality~\eqref{eq:LTtype} is the semi-classical one. In particular, we have 
\[
    \frac12 \Tr_d(- \Delta_d G) + c_{\TF}(s) \Tr_d(G^{1 + \frac{2}{s}})  \ge c_\TF(d+s) \int_{\R^d} \rho_G^{1 + \frac{2}{d+s}}.
\]
In other words, the energy in the rHF model is always greater than the energy in the TF model.



\end{document}